\documentclass[runningheads]{llncs}
\newif\ifdraft\draftfalse

\newif\iffull\fulltrue 

\usepackage{adjustbox}
\usepackage{xcolor}
\usepackage[colorlinks,linkcolor={red!50!black},
    citecolor={blue!50!black},
    urlcolor={blue!80!black}]{hyperref}
\usepackage{array}
\usepackage{graphicx}
\usepackage{amsfonts} 
\usepackage{amsmath}
\usepackage{amssymb}
\usepackage{booktabs}
\usepackage{cleveref}
\usepackage{csvsimple}
\usepackage{algorithm}
\usepackage{algpseudocode}
\algtext*{EndWhile}
\algtext*{EndIf}
\usepackage{todonotes}
\usepackage[framemethod=tikz]{mdframed}
\usepackage{stmaryrd}
\usepackage{cite}
\usepackage{colonequals}
\usepackage{subcaption}
\usepackage{microtype}
\usepackage{pgfplots}
\usepgfplotslibrary{external}
\usepackage{wrapfig}

\usepackage{amsmath}
\usepackage{nicematrix}

\makeatletter
\newcommand{\printfnsymbol}[1]{%
  \textsuperscript{\@fnsymbol{#1}}%
}
\makeatother

\usepackage{tikz}
\usetikzlibrary{calc}
\usetikzlibrary {arrows.meta,automata,positioning}
\usetikzlibrary{plotmarks}
\usepackage{tikz-cd}
\usepackage[all]{xy} 
\tikzstyle{point}=[circle,  inner sep=2pt, fill]

\newcommand{\scatterplotsize}[0]{0.9\textwidth}
\newcommand{\marksize}{1.8}

\newcommand{\scatterplot}[5]{%
	\begin{tikzpicture}
	\begin{axis}[
	width=\scatterplotsize,
	height=\scatterplotsize,
	axis equal image,
	xmin=1,
	ymin=1,
	ymax=9000,
	xmax=9000,
	xmode=log,
	ymode=log,
	axis x line=bottom,
	axis y line=left,
	xtick={1,9,90,900},
	xticklabels={1,9,90,900},
	extra x ticks = {2700, 9000},
	extra x tick labels = {OoR, IMP},
	extra x tick style = {grid = major},
	ytick={1,9,90,900},
	yticklabels={1,9,90,900},
	extra y ticks = {2700, 9000},
	extra y tick labels = {OoR, IMP},
	extra y tick style = {grid = major},
	xlabel={\scriptsize #3},
	xlabel style={yshift=0cm},
	ylabel={\scriptsize #5},
	ylabel style={yshift=-0.4cm},
	yticklabel style={font=\tiny},
	xticklabel style={rotate=290,anchor=west,font=\tiny},
 legend pos=outer north east,
	]

	\addplot[
	scatter,
	only marks,
	mark=asterisk,
 mark size=\marksize,
 fill opacity=0.5,
	scatter/use mapped color={
    draw=black,
    fill=black,
	}
	]%
	table [col sep=comma,x=#2,y=#4] {#1_7.txt};
	\addplot[
	scatter,
	only marks,
	mark=triangle,
 mark size=\marksize,
 fill opacity=0.5,
	scatter/use mapped color={
    draw=black,
    fill=black,
	}
	]%
	table [col sep=comma,x=#2,y=#4] {#1_9.txt};
	\addplot[
	scatter,
	only marks,
	mark=square,
 mark size=\marksize,
 fill opacity=0.5,
	scatter/use mapped color={
    draw=black,
    fill=black,
	}
	]%
	table [col sep=comma,x=#2,y=#4] {#1_8.txt};
	\addplot[
	scatter,
	only marks,
	mark=diamond,
 mark size=\marksize,
 fill opacity=0.5,
	scatter/use mapped color={
    draw=black,
    fill=black,
	}
	]%
	table [col sep=comma,x=#2,y=#4] {#1_1.txt};
	\addplot[
	scatter,
	only marks,
	mark=-,
 mark size=\marksize,
 fill opacity=0.5,
	scatter/use mapped color={
    draw=black,
    fill=black,
	}
	]%
	table [col sep=comma,x=#2,y=#4] {#1_2.txt};
	\addplot[
	scatter,
	only marks,
	mark=Mercedes star,
 mark size=\marksize,
 fill opacity=0.5,
	scatter/use mapped color={
    draw=black,
    fill=black,
	}
	]%
	table [col sep=comma,x=#2,y=#4] {#1_3.txt};
	\addplot[
	scatter,
	only marks,
	mark=Mercedes star flipped,
 mark size=\marksize,
 fill opacity=0.5,
	scatter/use mapped color={
    draw=black,
    fill=black,
	}
	]%
	table [col sep=comma,x=#2,y=#4] {#1_6.txt};
	\addplot[
	scatter,
	only marks,
	mark=heart,
 mark size=\marksize,
 fill opacity=0.5,
	scatter/use mapped color={
    draw=black,
    fill=black,
	}
	]%
	table [col sep=comma,x=#2,y=#4] {#1_4.txt};
	\addplot[
	scatter,
	only marks,
	mark=pentagon,
 mark size=\marksize,
 fill opacity=0.5,
	scatter/use mapped color={
    draw=black,
    fill=black,
	}
	]%
	table [col sep=comma,x=#2,y=#4] {#1_5.txt};

\legend{
}
	
	\addplot[no marks] coordinates {(0.01,0.01) (9000,9000) };
	\addplot[no marks, densely dotted] coordinates {(0.01,0.1) (900,9000)};
	\addplot[no marks, densely dotted] coordinates {(0.1,0.01) (9000,900)};
	\end{axis}
	\end{tikzpicture}
}

\algdef{SE}[DOWHILE]{Do}{doWhile}{\algorithmicdo}[1]{\algorithmicwhile\ #1}%
\algnewcommand\algorithmicforeach{\textbf{for each}}
\algdef{S}[FOR]{ForEach}[1]{\algorithmicforeach\ #1\ \algorithmicdo}

\renewcommand{\paragraph}[1]{\medskip\noindent\emph{#1}}

\renewcommand{\subsubsection}[1]{\medskip\noindent\textbf{#1}}

\DeclareMathOperator{\supp}{supp}

\newcommand{\littletaller}{\mathchoice{\vphantom{\big|}}{}{}{}}

\newcommand\restr[2]{{
  \left.\kern-\nulldelimiterspace 
  #1 
  \littletaller 
  \right|_{#2} 
  }}

\newcommand{\myparagraph}[1]{\smallskip\noindent \emph{#1}}

\newcommand{\defeq}{\colonequals}
\newcommand{\nat}{\mathbb{N}}
\newcommand{\dr}{\mathbf{r}}
\newcommand{\dl}{\mathbf{l}}
\newcommand{\seqcomp}{\fatsemi}

\newcommand{\interface}{\mathsf{IO}}
\newcommand{\id}{\mathrm{id}}

\newcommand{\semFunctor}{\mathcal{S}}

\newcommand{\real}{\mathbb R}

\newcommand{\sdist}[1]{\mathcal{D}_{\leq 1}(#1)}
\newcommand{\dist}[1]{\mathcal{D}(#1)}

\newcommand{\MaxReach}[3]{\mathrm{RPr}^{#1}_\text{max}(#2, #3)}

\newcommand{\Reacha}[3]{\mathrm{RPr}^{#1}(#2, #3)}
\newcommand{\WReacha}[3]{\mathrm{WRPr}^{#1}(#2, #3)}

\newcommand{\probinterval}{[0, 1]}
\newcommand{\Scheds}[1]{\mathrm{Sched}(#1)}

\newcommand{\sd}[1]{\mathbb{#1}}

\newcommand{\typemdp}[1]{\mathrm{tp}(#1)}
\newcommand{\semantics}[1]{\llbracket #1 \rrbracket}
\newcommand{\weights}[1]{\mathbf{W}^{#1}}

\newcommand{\undefined}{\bot}
\newcommand{\initstate}{s_\iota}
\newcommand{\Act}{A}
\newcommand{\enabled}[1]{A(#1)}
\newcommand{\mdppath}{\pi}
\newcommand{\last}[1]{\mathsf{last}(#1)}
\newcommand{\FinPaths}[1]{\mathsf{FPath}_{#1}}
\newcommand{\InfPaths}[1]{\mathsf{IPath}_{#1}}
\newcommand{\prmeas}[3]{\mathsf{Pr}_{#1}^{#2,#3}}

\newcommand{\Gmsched}[1]{\Sigma_\text{h}^{#1}}
\newcommand{\Possched}[1]{\Sigma_\text{sm}^{#1}}
\newcommand{\Dmsched}[1]{\Sigma_\text{dm}^{#1}}

\newcommand{\point}{\mathbf{p}}
\newcommand{\achievable}[2]{\mathsf{Ach}_{#1}(#2)}
\newcommand{\fachievable}[1]{\mathsf{Ach}_{#1}}

\newcommand{\achievablesched}[3]{\mathsf{Ach}^{#3}_{#1}(#2)}
\newcommand{\familyachievablesched}[2]{\mathsf{Ach}^{#2}_{#1}}
\newcommand{\convcl}[1]{\mathsf{ConvCl}(#1)}
\newcommand{\sybgap}{\mathsf{E}}
\newcommand{\gap}[2]{\sybgap(#1, #2)}
\newcommand{\gapinf}[2]{\sybgap_{\infty}(#1, #2)}
\newcommand{\dw}[1]{\mathsf{DwCl}(#1)}
\newcommand{\dwconvcl}[1]{\mathsf{DwConvCl}(#1)}
\newcommand{\indMC}[2]{#1[#2]}

\newcommand{\nxt}{\mathsf{Nx}}
\newcommand{\sched}{\sigma}

\newcommand{\mdp}[1]{\mathcal{#1}}
\newcommand{\cmdp}[1]{\mathcal{C}(#1)}

\newcommand{\arr}[1]{#1_{\dr}}
\newcommand{\arl}[1]{#1_{\dl}}

\newcommand{\convw}[1]{\mathbf{#1}}
\newcommand{\en}{I}
\newcommand{\ex}{O}
\newcommand{\enr}{I_{\dr}}
\newcommand{\enrarg}[1]{I_{\dr, #1}}
\newcommand{\enl}{I_{\dl}}
\newcommand{\enlarg}[1]{I_{\dl, #1}}
\newcommand{\exr}{O_{\dr}}
\newcommand{\exrarg}[1]{O_{\dr, #1}}
\newcommand{\exl}{O_{\dl}}
\newcommand{\exlarg}[1]{O_{\dl, #1}}

\renewcommand{\cref}[1]{\Cref{#1}}
\creflabelformat{enumi}{#2(#1)#3}
\crefname{theorem}{Thm.}{Theorems}
\crefname{mydefinition}{Def.}{Defs}
\crefname{proposition}{Prop.}{Props}
\crefname{lemma}{Lem.}{Lemmas}
\crefname{proof}{Proof.}{Proofs}
\crefname{appendix}{Append.}{Appendixes}
\crefname{example}{Ex.}{Examples}
\crefformat{section}{{\S}#2#1#3}
\crefname{figure}{Fig.}{Figs}
\Crefname{equation}{}{}

\newenvironment{proofs}{%
  \proof}{\endproof}

\begin{document}

\title{Pareto Curves for Compositionally Model Checking String Diagrams of MDPs\thanks{I.H.\ and K.W.\ are supported by ERATO HASUO Metamathematics for Systems Design Project (No.\ JPMJER1603) and the ASPIRE grant No.\ JPMJAP2301, JST. K.W.\ is supported by the JST grants No.\ JPMJFS2136 and  JPMJAX23CU. J.R.\ is supported by the NWO grant No.\ OCENW.M20.053.}}

\titlerunning{Pareto Curves for Compositionally Model Checking (...)  MDPs}
\author{Kazuki Watanabe\inst{1,2,}\thanks{Equal contribution.}\and Marck van der Vegt\inst{3,}\printfnsymbol{2} \and Ichiro Hasuo\inst{1,2} \and \\ Jurriaan Rot\inst{3} \and Sebastian Junges\inst{3}}
\institute{National Institute of Informatics, Tokyo, Japan\\ \email{\{kazukiwatanabe,hasuo\}@nii.ac.jp} \and The Graduate University for Advanced Studies (SOKENDAI), Hayama, Japan  \and Radboud University, Nijmegen, the Netherlands\\ \email{\{marck.vandervegt,sebastian.junges\}@ru.nl, jrot@cs.ru.nl}}
%

%
 \authorrunning{Watanabe et al.}
%
\maketitle              
\begin{abstract}
Computing schedulers that optimize reachability probabilities in MDPs is a standard verification task.
To address scalability concerns, we focus on MDPs that are compositionally described in a high-level description formalism.
In particular, this paper considers \emph{string diagrams}, which specify an algebraic, sequential composition of subMDPs.
Towards their compositional verification, the key challenge is to locally optimize schedulers on subMDPs without considering their context in the string diagram. 
This paper proposes to consider the schedulers in a subMDP which form a \emph{Pareto curve} on a combination of local objectives. While considering all such schedulers is intractable, it gives rise to a highly efficient sound approximation algorithm. 
The prototype on top of the model checker Storm demonstrates the scalability of this approach. 
\end{abstract}
%
\section{Introduction}
\label{sec:intro}
Markov decision processes (MDPs) are a ubiquitous model for describing systems with both nondeterministic and probabilistic uncertainty.
A key problem is to compute the best-case probability of reaching a goal state in a given MDP, i.e., to compute maximal reachability probabilities. Reachability probabilities can efficiently be computed for MDPs with $\approx 10^{7}$ states~\cite{DBLP:series/lncs/BaierHK19,DBLP:conf/tacas/HartmannsJQW23}, using mature model checkers such as \textsc{Storm}~\cite{HenselJKQV22}, PRISM~\cite{KwiatkowskaNP11} or Modest~\cite{HartmannsH14}. However, scalability beyond  state space sizes suffers from the memory limitations inherent to explicitly storing the transition matrix. While decision diagrams~\cite{DBLP:conf/icalp/BaierCHKR97,DBLP:conf/tacas/AlfaroKNPS00,DBLP:conf/cav/HoltzenJVMSB20} are powerful, compact representations, they fail to concisely represent many MDPs~\cite{DBLP:conf/isola/BuddeHKKPQTZ20}. 

\myparagraph{Sequential composition.} Compositional techniques attempt to avoid reasoning on the complete state space. We distinguish \emph{parallel} and \emph{sequential} compositions. 
This paper considers the compositional analysis of sequentially composed models~\cite{DBLP:conf/ijcai/BarryKL11}. This type of compositionality allows to reduce the peak memory consumption by reasoning about the individual parts and allows to exploit the typical existence of isomorphic parts of the state space. Sequentially composed MDPs have seen a surge in interest recently~\cite{Watanabe21,DBLP:conf/nips/JothimuruganBBA21,DBLP:conf/aips/NearyVCT22,JungesS22,WatanabeEAH23}. 

\myparagraph{String diagrams.} 
We focus on \emph{string diagrams of MDPs}~\cite{WatanabeEAH23}, which are MDPs composed by two algebraic operations: the \emph{sequential composition} $\seqcomp$ and the \emph{sum} $\oplus$. More precisely, we use \emph{open MDPs}, extending MDPs with entrance and exit states. \cref{fig:openMDPs} shows open MDPs $\mdp{A}$ (left) and $\mdp{B}$ (right). The open MDP $\mdp{A}$, for instance, has two entrances $i_{\dr, 1}, i_{\dl,1}$ and two exits $o_{\dr,1}, o_{\dl, 1}$.  The algebraic operations define how the open MDPs are \emph{subMDPs of a larger, monolithic MDP}, cf.\ \cref{fig:seqcompAndSum}. We highlight that, in our \emph{bi-directional} framework, the sequential composition of acyclic open MDPs may lead to a cyclic monolithic MDP.


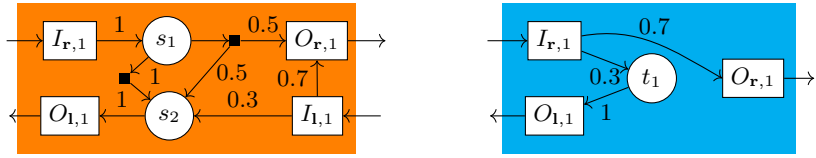
\begin{figure}[t]
\centering
\begin{tikzpicture}[
innode/.style={draw, rectangle, minimum size=0.5cm},
interface/.style={draw, rectangle, minimum size=0.5cm},]
\fill[orange] (-0.7cm, -1cm)--(-0.7cm, 1cm)--(3.8cm, 1cm)--(3.8cm, -1cm)--cycle;
\node[interface,fill=white,yshift=0.5cm] (s0) {$\enrarg{1}$};
\node[inner sep=0,right=-1.25cm of s0] (enr1) {};
\node[interface,fill=white,yshift=-0.5cm] (s0o) {$\exlarg{1}$};
\node[inner sep=0,right=-1.25cm of s0o] (exl1) {};
\node[state,right=0.6cm of s0, minimum size=0.5cm,fill=white] (s1) {$s_1$};
\node[state,right=0.6cm of s0o, minimum size=0.5cm,fill=white] (ssink) {$s_2$};

\node[inner sep=2pt, fill=black,right=0.5cm of s1] (a1a) {};
\node[inner sep=2pt, fill=black,right=-1cm of ssink,yshift=0.5cm] (a1b) {};

\node[interface, right=1.3cm of ssink,fill=white] (s2) {$\enlarg{1}$};
\node[interface, right=0.6cm of a1a,fill=white] (s3) {$\exrarg{1}$};
\node[inner sep=0,right= 0.5cm of s2] (exr1) {};
\node[inner sep=0,right= 0.5cm of s3] (exr2) {};
\draw[->] (enr1) -> (s0);
\draw[->] (s0o) -> (exl1);
\draw[->] (s0) -> node [above] {$1$} (s1);
\draw[->] (s1) -> node [above] {} (a1a);
\draw[->] (s1) -> node [above] {} (a1b);
\draw[->] (a1a) -> node [above] {$0.5$} (s3);
\draw[->] (a1a) -> node [right] {$0.5$} (ssink);
\draw[->] (ssink) -> node [above] {$1$} (s0o);
\draw[->] (a1b) -> node [right,yshift=0.2cm] {$1$} (ssink);
\draw[->] (exr1) -> (s2);
\draw[->] (s3) -> (exr2);
\draw[->] (s2) -> node [above] {$0.3$} (ssink);
\draw[->] (s2) -> node [left] {$0.7$} (s3);
\end{tikzpicture}
\hspace{30pt}
\begin{tikzpicture}[
innode/.style={draw, rectangle, minimum size=0.5cm},
interface/.style={draw, rectangle, minimum size=0.5cm},]
\fill[cyan] (-0.7cm, -1cm)--(-0.7cm, 1cm)--(3.2cm, 1cm)--(3.2cm, -1cm)--cycle;
\node[interface, yshift=0.5cm,fill=white] (t0) {$\enrarg{1}$};
\node[interface, yshift=-0.5cm,fill=white] (t1) {$\exlarg{1}$};
\node[inner sep=0,right=-1.3cm of t0] (enr1) {};
\node[inner sep=0,right=-1.3cm of t1] (enr2) {};
\node[state,right=0.6cm of t0, yshift=-0.5cm, minimum size=0.5cm,fill=white] (t2) {$t_1$};
\node[interface, right=0.6cm of t2,fill=white] (t3) {$\exrarg{1}$};
\node[inner sep=0,right=0.4cm of t3] (exr1) {};

\draw[->] (enr1) -> (t0);
\draw[->] (t1) -> (enr2);
\draw[->] (t0) -> node [below] {$0.3$} (t2);
\draw[->] (t2) -> node [below] {$1$} (t1);
\draw[->] (t0) to [out=10,in=150] node [above] {$0.7$} (t3.west);
\draw[->] (t3) -> (exr1);
\end{tikzpicture}
\caption{Open Markov decision processes $\mdp{A}$ and $\mdp{B}$.  }
\label{fig:openMDPs}
\end{figure}
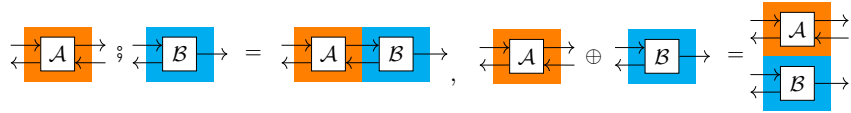
\begin{figure}[t]
\centering
\vspace{-3mm}
\begin{minipage}[]{0.45\linewidth}
\scalebox{0.9}{
\begin{tikzpicture}[
innode/.style={draw, rectangle, minimum size=0.5cm},
interface/.style={inner sep=0},]
\fill[orange] (-0.5cm, -0.4cm)--(-0.5cm, 0.4cm)--(0.5cm, 0.4cm)--(0.5cm, -0.4cm)--cycle;
\node[innode,fill=white] (mdpA) {$\mdp{A}$};
\node[inner sep=0,right=-1cm of mdpA,yshift=0.13cm] (enrA1) {};
\node[inner sep=0,right=-1cm of mdpA,yshift=-0.13cm] (exlA1) {};
\node[inner sep=0,right=0.45cm of mdpA,yshift=0.13cm] (exrA1) {};
\node[inner sep=0,right=0.45cm of mdpA,yshift=-0.13cm] (exrA2) {};
\draw[->] (enrA1) -> ($(mdpA.west)!0.5!(mdpA.north west)$);
\draw[->] ($(mdpA.west)!0.5!(mdpA.south west)$) -> (exlA1);
\draw[->] ($(mdpA.east)!0.5!(mdpA.north east)$) -> (exrA1);
\draw[<-] ($(mdpA.east)!0.5!(mdpA.south east)$) -> (exrA2);
\node[interface, right=0.1cm of exrA1, yshift=-0.13cm] (seqcomp) {$\seqcomp$};
\fill[cyan] (1.3cm, -0.4cm)--(1.3cm, 0.4cm)--(2.3cm, 0.4cm)--(2.3cm, -0.4cm)--cycle;
\node[innode,fill=white] (mdpB) at (1.8cm, 0cm) {$\mdp{B}$};
\node[inner sep=0,right=-1cm of mdpB,yshift=0.13cm] (enrB1) {};
\node[inner sep=0,right=-1cm of mdpB,yshift=-0.13cm] (enrB2) {};
\node[inner sep=0,right=0.45cm of mdpB] (exrB1) {};
\draw[->] (enrB1) -> ($(mdpB.west)!0.5!(mdpB.north west)$);
\draw[<-] (enrB2) -> ($(mdpB.west)!0.5!(mdpB.south west)$);
\node[inner sep=0,right=0.45cm of mdpB] (exrB) {};
\draw[->] (mdpB) -> (exrB);
\node[inner sep=0,right=0.2cm of exrB] (eqseqcomp) {$=$};
\fill[orange] (3.5cm, -0.4cm)--(3.5cm, 0.4cm)--(4.5cm, 0.4cm)--(4.5cm, -0.4cm)--cycle;
\node[innode, fill=white] (mdpA2) at (4cm, 0cm) {$\mdp{A}$};
\node[inner sep=0,right=-1cm of mdpA2,yshift=0.13cm] (enrA2) {};
\node[inner sep=0,right=-1cm of mdpA2,yshift=-0.13cm] (exlA2) {};
\fill[cyan] (4.5cm, -0.4cm)--(4.5cm, 0.4cm)--(5.5cm, 0.4cm)--(5.5cm, -0.4cm)--cycle;
\node[innode, fill=white] (mdpB2) at (5cm, 0cm) {$\mdp{B}$};
\draw[->] (enrA2) -> ($(mdpA2.west)!0.5!(mdpA2.north west)$);
\draw[<-] (exlA2) -> ($(mdpA2.west)!0.5!(mdpA2.south west)$);
\draw[->] ($(mdpA2.east)!0.5!(mdpA2.north east)$) -> ($(mdpB2.west)!0.5!(mdpB2.north west)$);
\draw[<-] ($(mdpA2.east)!0.5!(mdpA2.south east)$) -> ($(mdpB2.west)!0.5!(mdpB2.south west)$);
\node[inner sep=0,right=0.5cm of mdpB2] (exrB2) {};
\draw[->] (mdpB2) -> (exrB2);
\end{tikzpicture}},
\end{minipage}
\hspace{15pt}
\begin{minipage}[]{0.45\linewidth}
\scalebox{0.9}{
\begin{tikzpicture}[
innode/.style={draw, rectangle, minimum size=0.5cm},
interface/.style={inner sep=0},]
\fill[orange] (-0.5cm, -0.4cm)--(-0.5cm, 0.4cm)--(0.5cm, 0.4cm)--(0.5cm, -0.4cm)--cycle;
\node[innode,fill=white] (mdpA) {$\mdp{A}$};
\node[inner sep=0,right=-1cm of mdpA,yshift=0.13cm] (enrA1) {};
\node[inner sep=0,right=-1cm of mdpA,yshift=-0.13cm] (exlA1) {};
\node[inner sep=0,right=0.45cm of mdpA,yshift=0.13cm] (exrA1) {};
\node[inner sep=0,right=0.45cm of mdpA,yshift=-0.13cm] (exrA2) {};
\draw[->] (enrA1) -> ($(mdpA.west)!0.5!(mdpA.north west)$);
\draw[<-] (exlA1) -> ($(mdpA.west)!0.5!(mdpA.south west)$);
\draw[->] ($(mdpA.east)!0.5!(mdpA.north east)$) -> (exrA1);
\draw[<-] ($(mdpA.east)!0.5!(mdpA.south east)$) -> (exrA2);
\node[interface, right=0.1cm of exrA1, yshift=-0.13cm] (oplus) {$\oplus$};
\fill[cyan] (1.5cm, -0.4cm)--(1.5cm, 0.4cm)--(2.5cm, 0.4cm)--(2.5cm, -0.4cm)--cycle;
\node[innode,fill=white] (mdpB) at (2cm, 0cm) {$\mdp{B}$};
\node[inner sep=0,right=-1cm of mdpB,yshift=0.13cm] (enrB1) {};
\node[inner sep=0,right=-1cm of mdpB,yshift=-0.13cm] (enrB2) {};
\node[inner sep=0,right=0.45cm of mdpB] (exrB1) {};
\draw[->] (enrB1) -> ($(mdpB.west)!0.5!(mdpB.north west)$);
\draw[<-] (enrB2) -> ($(mdpB.west)!0.5!(mdpB.south west)$);
\node[inner sep=0,right=0.45cm of mdpB] (exrB) {};
\draw[->] (mdpB) -> (exrB);
\node[inner sep=0,right=0.2cm of exrB] (eqseqcomp) {$=$};
\fill[orange] (3.5cm, 0cm)--(3.5cm, 0.8cm)--(4.5cm, 0.8cm)--(4.5cm, 0cm)--cycle;
\node[innode, fill=white] (mdpA2) at (4cm, 0.4cm) {$\mdp{A}$};
\node[inner sep=0,right=-1cm of mdpA2,yshift=0.13cm] (enrA2) {};
\node[inner sep=0,right=-1cm of mdpA2,yshift=-0.13cm] (exlA2) {};
\fill[cyan] (3.5cm, 0cm)--(3.5cm, -0.8cm)--(4.5cm, -0.8cm)--(4.5cm, 0cm)--cycle;
\node[innode, fill=white] (mdpB2) at (4cm, -0.4cm) {$\mdp{B}$};
\draw[->] (enrA2) -> ($(mdpA2.west)!0.5!(mdpA2.north west)$);
\draw[<-] (exlA2) -> ($(mdpA2.west)!0.5!(mdpA2.south west)$);
\node[inner sep=0,right=0.5cm of mdpA2,yshift=0.13cm] (exrA21) {};
\node[inner sep=0,right=0.5cm of mdpA2,yshift=-0.13cm] (exrA22) {};
\node[inner sep=0,right=-1cm of mdpB2,yshift=0.13cm] (enrB21) {};
\node[inner sep=0,right=-1cm of mdpB2,yshift=-0.13cm] (enrB22) {};
\draw[->] ($(mdpA2.east)!0.5!(mdpA2.north east)$) -> (exrA21);
\draw[<-] ($(mdpA2.east)!0.5!(mdpA2.south east)$) -> (exrA22);
\draw[->] (enrB21) -> ($(mdpB2.west)!0.5!(mdpB2.north west)$);
\draw[<-] (enrB22) -> ($(mdpB2.west)!0.5!(mdpB2.south west)$);
\node[inner sep=0,right=0.5cm of mdpB2] (exrB2) {};
\draw[->] (mdpB2) -> (exrB2);
\end{tikzpicture}
}
\end{minipage}
\caption{Sequential composition $\mdp{A}\seqcomp\mdp{B}$ and sum $\mdp{A}\oplus\mdp{B}$}
\label{fig:seqcompAndSum}
\end{figure}

\myparagraph{Optimal local schedulers.} 
The idea of compositional reasoning is to analyze the open MDPs individually and combine these results to answer reachability queries on the monolithic MDP. The key challenge is that during the analysis of an individual open MDP, it is unclear which exits are (un)desirable to be reached. 
Equivalently, a priori, we do not know the objective that a scheduler should (locally) optimize for, to resolve the nondeterminism in the open MDP.

\myparagraph{State-of-the-art.}
So far, this problem has been circumvented in the literature on compositional MDP verification. 
In~\cite{JungesS22}, the notion of locally optimal policies is used. In essence, the technique relies on syntactic restrictions, such as open MDPs without nondeterminism or with a dedicated \emph{desirable} exit.
In~\cite{DBLP:conf/aips/NearyVCT22}, it is assumed that agents must optimize to reach one of the exits and that reaching another exit is equivalent to reaching an error.
The results in~\cite{WatanabeEAH23} are the first to support general string diagrams of MDPs. Algorithmically, they enumerate over all (deterministic memoryless) schedulers in every open MDP, reducing the resulting set of schedulers only a posteriori using so-called meagre semantics.
Finally, work on compositional \emph{planning} aims to find a good schedulers compositionally, but without any guarantees on the optimality~\cite{DBLP:conf/ijcai/BarryKL11}. 

\myparagraph{A multi-objective perspective.}
Towards a compositional analysis, we reformulate and generalize our problem slightly. Rather than considering \emph{Given an open MDP, what is the maximal probability to reach a dedicated exit?}, we maximize the probability towards each exit individually, i.e., \emph{What is the maximal probability to reach the first and second exit, respectively?}. For this question, trade-offs are possible: In $\mdp{A}$ when starting in $I_{r,1}$ we either reach $O_{l,1}$ with probability $1$ (with one scheduler) or $O_{r,1}$ with probability $0.5$ (with another scheduler). However, we search for \emph{one scheduler} that makes this trade-off somehow. The unknown objective that a  scheduler in an open MDP optimizes for is not arbitrary, but it is given by this trade-off between reaching the different exits. A key insight is that it suffices to only consider schedulers that refer to an optimal trade-off between the different exits. As the context of the open MDP determines the preferred trade-off, we compute all schedulers that are optimal for a specific trade-off. These schedulers are \emph{Pareto-optimal} and their computation is well-studied~\cite{EtessamiKVY08,ForejtKP12,QuatmannK21}.

\myparagraph{Our approach.}
Towards a compositional algorithm, we suggest to compute Pareto curves recursively on the structure of the given string diagram of MDPs: given the Pareto curves for the open MDPs, we can compute the Pareto curve of their composition. 
As the set of Pareto-optimal schedulers remains exponential, we exploit efficient but approximative approaches to compute sound over-and-under approximations~\cite{ForejtKP12}. In practice, this means that tight approximations can be achieved that are concisely represented using only a  few schedulers. Given these sound approximations for each open MDP, we compute sound approximations\footnote{Sound approximations are  standard in probabilistic model checking, where the standard and highly scalable algorithms~\cite{DBLP:conf/cav/HartmannsK20} provide sound approximations~\cite{DBLP:conf/tacas/HartmannsJQW23}.} for their composition, and ultimately for the whole string diagram.


\myparagraph{Contributions}
Our technical contribution is as follows. We provide a novel  framework for analysing sequentially composed MDPs. 
In particular, it takes off-the-shelf analysis of \emph{multi-objective monolithic MDPs} to provide an \emph{compositional MDP model checking algorithm}. The approximative version of this algorithm computes guaranteed over- and under-approximations of reachability probabilities and scales to models with both billions of states and schedulers, while generating tight bounds.
We implement the algorithm on top of the probabilistic model checker \textsc{Storm} to demonstrate its performance.

\section{Overview}
\label{sec:overview}

\begin{figure}[t]
\centering
\includegraphics[scale=0.33]{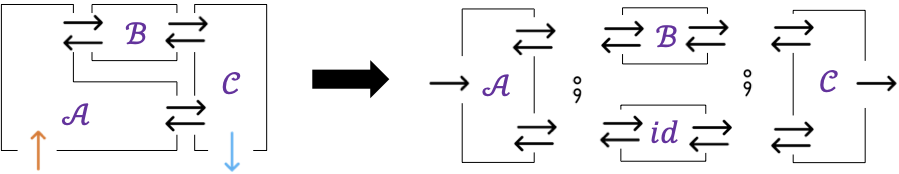}
\caption{From a multi-room MDP to a string diagram $\mdp{A}\seqcomp (\mdp{B} \oplus \id)\seqcomp \mdp{C}$.} 
\label{fig:motivatingexample}
\end{figure}

\myparagraph{Illustrative example.} We consider the small multi-room grid world with three rooms $\mdp{A}, \mdp{B}, \mdp{C}$ and probabilistic outcomes to actions, as illustrated on the left in \cref{fig:motivatingexample}. The doors can be travelled through in both directions. The grid world can be modelled as a monolithic MDP. 
In \cref{fig:motivatingexample}(right), we show how to express the MDP compositionally as a string diagram $\mdp{A}\seqcomp (\mdp{B} \oplus \id) \seqcomp \mdp{C}$, where $\id$ is an MDP where the unique entrance reaches the exit with probability one. Now, if we are interested (how to) reach the (main, rightmost) exit in $\mdp{C}$ from the (main, leftmost) entrance in $\mdp{A}$, we maximize the reachability probability in the monolithic MDP. However, to determine the optimal scheduler compositionally, we must know the optimal reachabilities in the each room individually. This is hard: In particular, it may be true that in order to reach the main exit from the door between $\mdp{A}$ and $\mdp{C}$, it is still optimal to go via room $\mdp{B}$.  

\myparagraph{The multiobjective perspective.}
We consider room $\mdp{A}$ from the perspective of the main entrance. There are two doors, which means that the underlying open MDP has two exit states. \cref{subfig:ExPareto} shows a Pareto-plot that clarifies the possible trade-offs, e.g., for the main entrance\footnote{An actual MDP corresponding to this Pareto curve is given in \cref{subfig:ExOpenMDP}}. In particular, $\point_1$ reflects a scheduler reaching the first door with probability $0.3$, while the second door is then reached with probability $0.1$. The other points reflect other schedulers that reach these doors with different probabilities.
The other entrances in $\mdp{A}$ induce other Pareto curves. In room $\mdp{C}$, there are three exits (two doors and the main exit), making the Pareto curves three-dimensional. For MDP $\id$, the curve is a (trivial) point.

\begin{wrapfigure}{r}[0pt]{0.33\textwidth}
\centering
\vspace{-5mm}
\includegraphics[scale=0.25]{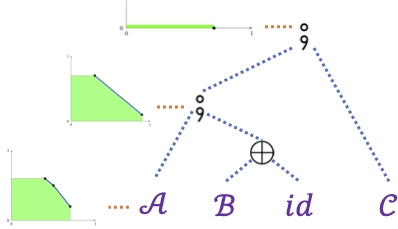}
\caption{Our approach}
\label{fig:syntaxtree}
\vspace{-7mm}
\end{wrapfigure}

\myparagraph{The approach illustrated.}
We approach the syntax tree of the string diagram recursively (see \cref{fig:syntaxtree}), i.e., we consider the string diagram as a syntax tree and (conceptually) annotate the MDPs $\mdp{A}$, $\mdp{B}$, $\mdp{C}$ and $\id$ with Pareto curves.
Instead of computing these Pareto curves precisely, we create sound approximations as in \cref{subfig:AxisPareto} where the vertices of the green area $L$ underapproximate the Pareto curve and the vertices of the green and white area $\mathbb{R}^2 \setminus U$ overapproximate the Pareto-curve. For the algebraic operations, we then combine these approximations. 
This is straightforward for the $\oplus$ as the two MDPs are independent. For the sequential composition, the cyclic dependencies are more involved. However, our algorithm is straightforward: we take the individual approximations of the Pareto curves, translate them to small so-called \emph{shortcut MDPs}, which we then compose. On these small MDPs, we then compute (approximations of) the Pareto curves that are sound approximations to the composition at hand.


\section{Formal Problem Statement}
\label{sec:problem}

We recap MDPs, and discuss string diagrams of MDPs. We then give a multi-objective version of the problem statement, leading to a compositional algorithm. 



For a finite set $X$, we write $\dist{X}$ for the set of distributions on $X$ and $\sdist{X}$ for the subdistributions. 
The support of $\mu \in \sdist{X}$ is denoted by $\supp(\mu)$.

\subsection{Markov Decision Processes}
\begin{definition}[MDP]
An MDP $M = (S, A, P)$ is a tuple with a finite set $S$ of \emph{states}, finite set $\Act$ of \emph{actions}, and a partial \emph{transition function} $P\colon S\times \Act \rightharpoonup \dist{S}$. 
\end{definition}
We use $S_M, \Act_M$ and $P_M$ to refer to the states, actions and transition function of an MDP $M$. For a state $s$, the \emph{enabled actions} are $\enabled{s} = \{ a \in \Act \mid P(s,a) \neq \undefined \}$. A state $s$ is a \emph{terminal} if $\enabled{s} = \emptyset$.
We write $P(s,a,s') \defeq P(s,a)(s')$ if $a \in \enabled{s}$ and $P(s,a,s') \defeq 0$ otherwise. 
A state $s$ is called \emph{absorbing} if $P(s,a,s) = 1$ for all $a \in \enabled{s}$. Terminals can be made absorbing by adding a self-loop.
A \emph{path} $\mdppath$ is an (in)finite alternating sequence of states and actions, i.e., $\mdppath = s_0\xrightarrow{a_0}s_1\xrightarrow{a_1}\dots$ such that $s_i \in S$ and $P(s_i, a_i, s_{i+1}) \neq 0$. 
For a finite path $\mdppath$, $\last{\mdppath}$ denotes the last state of $\mdppath$. 
The set of all finite paths is denoted $\FinPaths{M}$, the set of all infinite paths is denoted by $\InfPaths{M}$. 
We drop the subscript, whenever $M$ is clear from the context.
Finally, a \emph{Markov chain} (MC) is an MDP with $|\enabled{s}| \leq 1$. We write MCs as a tuple $(S,P)$, formally representing the MDP $(S, \{ \bot \}, P)$.

\paragraph{Schedulers.}
Schedulers (a.k.a.\ \emph{policies} or \emph{strategies}) resolve the nondeterminism in an MDP. In general, a (history-dependent) scheduler $\sched$ for MDP $M$ is a (measurable) 
function $\FinPaths{M}\rightarrow \dist{\Act_M}$ with $\supp(\sched(\mdppath)) \subseteq \enabled{\last{\mdppath}}$. 
The set of all history-dependent schedulers is denoted $\Gmsched{M}$. 
A scheduler $\sched$ is \emph{memoryless} (a.k.a. \emph{positional}) if for every $\mdppath,\mdppath' \in \FinPaths{}$,  $\last{\mdppath} = \last{\mdppath'}$ implies $\sched(\mdppath) = \sched(\mdppath')$. The set of all (stochastic) memoryless schedulers is denoted $\Possched{M}$. We often write such schedulers as $S_M \rightarrow \dist{\Act_M}$, i.e., as a map from the last state to a distribution over actions.
Finally, deterministic (memoryless) schedulers map paths (or states, respectively) to Dirac distributions. We write such schedulers as maps from paths to actions,  $\FinPaths{M}\rightarrow \Act_M$. The set of all deterministic memoryless schedulers is denoted $\Dmsched{M}$. We also call these schedulers DM schedulers.

\paragraph{Reachability probabilities.}
Let $M$ be an MDP without terminals.  For an initial state $\initstate$ and a scheduler $\sched$, we obtain probability measure
$\prmeas{M}{\initstate}{\sched} \colon \InfPaths{M} \rightarrow \mathbb{R}_{\geq 0}$
on infinite paths via the standard cylinder set construction~\cite{Baier08}. 
For a set of target states $T \subseteq S_M$, we define the set of paths that visit $T$, $\lozenge T = \{ \mdppath \in \InfPaths{M} \mid  \exists i. \mdppath_i \in T \}$. 
The reachability probability $\Reacha{M,\sched}{\initstate}{T}$ is the integral over the reachability measure, $\Reacha{M,\sched}{\initstate}{T} \defeq \int_{\mdppath \in {\lozenge {T}}} \prmeas{M}{\initstate}{\sched}(\mdppath)$. We relax notation and write $\Reacha{M,\sched}{\initstate}{t} \defeq \Reacha{M,\sched}{\initstate}{\{ t \}}$. We write $\MaxReach{M}{\initstate}{t} \defeq \max_{\sched} \Reacha{M,\sched}{\initstate}{t}$.

\subsection{String Diagrams of MDPs}
MDPs are given as a string diagram, i.e., as algebraically composed \emph{open MDPs}.
\begin{definition}[oMDP]
\label{def:openMDPs}
An \emph{open MDP} (oMDP) $\mdp{A} = (M, \interface)$ is a pair consisting of an MDP $M$ with states $S$ and \emph{open ends} $\interface =(\enr, \enl, \exr, \exl)$, where 
  $\enr, \enl, \exr, \exl\subseteq S$ are pairwise disjoint and totally ordered sets. The states $\en \defeq \enr\cup \enl$ is the \emph{entrances}, and the states  $\ex \defeq \exr\cup\exl$ is the \emph{exits}, respectively. 
\end{definition}
\cref{fig:openMDPs} shows two oMDPs as examples. 
We assume that exactly the exits are terminals.
We lift the notions of policies and reachability probabilities straightforwardly from MDPs\footnote{Where terminals have to be made absorbing by adding a self-loop.}. 
As the open ends are ordered, we may enumerate, e.g., its entrances $I$ using $\enrarg{1},\dots, \enrarg{|\enr|}, \enlarg{1},\dots, \enlarg{|\enl|}$ from \emph{rightward} to \emph{leftward}. We specialise the notation for reachability probabilities $\Reacha{*}{i}{j}$ to denote $\Reacha{*}{I_i}{O_j}$.
We explicitly write $\typemdp{\mdp{A}}\colon (\arr{m}, \arl{m})\rightarrow  (\arr{n}, \arl{n})$ for the \emph{arities} of $\mdp{A}$, where $\arr{m} \defeq |\enr|$, $\arl{m} \defeq |\exl|$, $\arr{n} \defeq |\exr|$, and $\arl{n} \defeq |\enl|$.

String diagrams of MDPs use two algebraic operations on oMDPs: the \emph{sequential composition $\seqcomp$} and the \emph{sum $\oplus$}, that we illustrated already in \cref{fig:seqcompAndSum}.



\begin{definition}[$\seqcomp$ operator]\label{def:seqOMDP}
    Let $\mdp{A}$, $\mdp{B}$ be oMDPs, $\typemdp{\mdp{A}} = (\arr{m}, \arl{m}) \rightarrow (\arr{l}, \arl{l})$, $\typemdp{\mdp{B}} = (\arr{l}, \arl{l}) \rightarrow (\arr{n}, \arl{n})$. Their \emph{sequential composition} $\mdp{A} \seqcomp \mdp{B}$  is an oMDP $(M, \interface')$ with  $\interface' = (\enr^{\mdp{A}}, \enl^{\mdp{B}}, \exr^{\mdp{B}}, \exl^{\mdp{A}})$,  $M \defeq (S^{\mdp{A}} \uplus S^{\mdp{B}}, A^{\mdp{A}} \uplus A^{\mdp{B}}, P)$ and $P$  s.t.\  
    \begin{align*}
        P(s, a, s') &\defeq \begin{cases}
            P^{\mdp{D}}(s, a, s') & \text{ if $\mdp{D} \in \{\mdp{A}, \mdp{B}\}$, $s\in S^{\mdp{D}}$, $a\in A^{\mdp{D}}$, $s'\in S^{\mdp{D}}$, }\\
            1 &\text{ if $s = \exrarg{i}^{\mdp{A}}$, $s' = \enrarg{i}^{\mdp{B}}$ for some $1 \leq i \leq \arr{l}$},\\
             1 &\text{ if $s = \exlarg{i}^{\mdp{B}}$, $s' = \enlarg{i}^{\mdp{A}}$ for some $1 \leq i \leq \arl{l}$},\\
            0 &\text{otherwise. }
        \end{cases}
    \end{align*}
\end{definition}
If $\mdp{A}\seqcomp\mdp{B}$ is well-defined, we say that $\mdp{A} \seqcomp \mdp{B}$ type-matches.

\begin{definition}[$\oplus$ operator]
\label{def:sumOMDP}
    Let $\mdp{A}, \mdp{B}$ be oMDPs.
    Their \emph{sum} $\mdp{A}\oplus\mdp{B}$ is an oMDP $(M, \interface')$ with  $\interface' = (\enr^{\mdp{A}}\uplus \enr^{\mdp{B}}, \enl^{\mdp{A}}\uplus \enl^{\mdp{B}}, \exr^{\mdp{A}}\uplus \exr^{\mdp{B}}, \exl^{\mdp{A}}\uplus \exl^{\mdp{B}})$ and  $M = (S^{\mdp{A}} \uplus S^{\mdp{B}}, A^{\mdp{A}} \uplus A^{\mdp{B}}, P)$ where $P$ is given by   
    \begin{align*}
        P(s, a, s') &\defeq \begin{cases}
            P^{\mdp{D}}(s, a, s') &\text{ if $\mdp{D} \in \{\mdp{A}, \mdp{B}\}$, $s\in S^{\mdp{D}}$, $a\in A^{\mdp{D}}$, and $s'\in S^{\mdp{D}}$, }\\
            0 &\text{otherwise. }
        \end{cases}
    \end{align*}
Here, the total order in $\enr^{\mdp{A}}\uplus \enr^{\mdp{B}}$ is given by $\enrarg{1}^{\mdp{A}},\dots, \enrarg{\arr{m}}^{\mdp{A}}, \enrarg{1}^{\mdp{B}},\dots, \enrarg{\arr{k}}^{\mdp{B}}$, and the total orders in other open ends are defined similarly. 
\end{definition}

\noindent Using the algebraic operators, we define \emph{string diagrams of MDPs}.
\begin{definition}
\label{def:sd_mdps}
A \emph{string diagram $\sd{D}$ of MDPs} is a term adhering to the  grammar
\[ \sd{D} ::= \mdp{A} \mid \sd{D} \seqcomp \sd{D} \mid \sd{D} \oplus \sd{D}\]
where $\mdp{A}$ ranges over oMDPs.
 The \emph{operational semantics} $\semantics{\sd{D}}$ is the oMDP which is inductively defined by Definitions~\ref{def:seqOMDP} and~\ref{def:sumOMDP}.
\end{definition}
Throughout this paper, we assume that $\semantics{\sd{D}}$ is well-defined, i.e., that all operations type-match. We omit syntactic sugar operations, such as  probabilistic or nondeterministic branching, as these can be modelled inside oMDPs. In the literature, the $\semantics{\sd{D}}$ are also  referred to as the \emph{monolithic MDP} for $\sd{D}$.

\begin{mdframed}
\textbf{Single-Exit Problem Statement:}
Given a string diagram $\sd{D}$, an entrance $i$, an exit $j$, and an error bound $\epsilon\in \probinterval$, compute a scheduler $\sched$ such that 
\begin{equation*}
     \MaxReach{\semantics{\sd{D}}}{i}{j} - \Reacha{\semantics{\sd{D}},\sched}{i}{j} \leq \epsilon.
\end{equation*}
\end{mdframed}
We remark that DM schedulers suffice for this problem statement.

\begin{remark}
\label{rem:difference_CAV23}
String diagrams are traditionally graphical languages based on category theory, which involve not only terms but also equations; see~\cite{MacLane2,selinger2011survey,hinze_marsden_2023}. 
The definition of string diagrams of MDPs in~\cite{WatanabeEAH23} follows in this tradition 
and satisfies certain equational axioms. 
In this paper, the equations do not play a role explicitly; our algorithms assume a syntactic presentation. Solely for the purpose of exposition, we use a more concise definition where some axioms do not hold (although they ``essentially'' hold, modulo isomorphisms and removing redundant open ends). All definitions, results, and even proofs in this paper are concretely described and self-contained, without any use of category theory. 
\end{remark}

\subsection{A Multi-Objective Generalization}
\label{subsec:multi-ob-generalization}
Towards a recursive, compositional formulation of the problem statement, we generalize the single-exit problem to allow for multiple exits. Concretely, the generalized problem statement is to compute \emph{Pareto curves}~\cite{PapadimitriouY00,ForejtKP12} that represent combinations of reachability probabilities towards a set of exits.

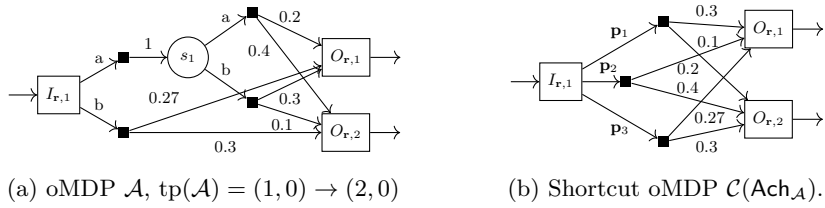
\begin{figure}[t]
    \centering
    \begin{minipage}[b]{0.45\hsize}
    \begin{tikzpicture}[
innode/.style={draw, rectangle, minimum size=0.5cm},
interface/.style={draw, rectangle, minimum size=0.5cm},]
        \node[interface,fill=white,yshift=1cm] (s0) {\scalebox{0.7}{$\enrarg{1}$}};
        \node[inner sep=0,right=-1cm of s0] (enr1) {};
        \node[inner sep=2pt, fill=black,right=0.5cm of s0, yshift=0.5cm] (as0) {};
        \node[inner sep=2pt, fill=black,right=0.5cm of s0, yshift=-0.5cm] (bs0) {};
        \node[state,right=0.5cm of as0, minimum size=0.5cm,fill=white] (s1) {\scalebox{0.7}{$s_1$}};
        \node[inner sep=2pt, fill=black,right=0.5cm of s1, yshift=0.6cm] (as1) {};
        \node[inner sep=2pt, fill=black,right=0.5cm of s1, yshift=-0.6cm] (bs1) {};
        \node[interface,right=1.5cm of s1,fill=white] (s3) {\scalebox{0.7}{$\exrarg{1}$}};
        \node[interface,right=-0.65cm of s3,fill=white, yshift=-1cm] (s4) {\scalebox{0.7}{$\exrarg{2}$}};
        \node[inner sep=0,right=0.4cm of s3] (exr1) {};
        \node[inner sep=0,right=0.4cm of s4] (exr2) {};
        \draw[->] (enr1) -> (s0);
        \draw[->] (s0) -> node [above] {$\scalebox{0.7}{a}$} (as0);
        \draw[->] (s0) -> node [above] {$\scalebox{0.7}{b}$} (bs0);
        \draw[->] (as0) -> node [above] {$\scalebox{0.7}{1}$} (s1);
        \draw[->] (s1) -> node [above] {$\scalebox{0.7}{a}$} (as1);
        \draw[->] (s1) -> node [above] {$\scalebox{0.7}{b}$} (bs1);
        \draw[->] (bs0) -> node [left, xshift=-0.5cm] {$\scalebox{0.7}{0.27}$} (s3);
        \draw[->] (bs0) -> node [below] {$\scalebox{0.7}{0.3}$} (s4);
        \draw[->] (as1) -> node [above] {$\scalebox{0.7}{0.2}$} (s3); 
        \draw[->] (as1) -> node [left, yshift=0.2cm, xshift=-0.2cm] {$\scalebox{0.7}{0.4}$} (s4); 
        \draw[->] (bs1) -> node [below] {$\scalebox{0.7}{0.3}$} (s3);
        \draw[->] (bs1) -> node [below,xshift=-0.1cm, yshift=0.05cm] {$\scalebox{0.7}{0.1}$} (s4);
        \draw[->] (s3) -> (exr1);
        \draw[->] (s4) -> (exr2);
    \end{tikzpicture}
    \centering
    \subcaption{oMDP $\mdp{A}$, $\typemdp{\mdp{A}} = (1, 0) \rightarrow (2, 0)$}\label{subfig:ExOpenMDP}
    \end{minipage}
    \begin{minipage}[b]{0.54\hsize}
    \centering
    \begin{tikzpicture}[
innode/.style={draw, rectangle, minimum size=0.5cm},
interface/.style={draw, rectangle, minimum size=0.5cm},]
        \node[interface,fill=white,yshift=1cm] (s0) {\scalebox{0.7}{$\enrarg{1}$}};
        \node[inner sep=0,right=-1cm of s0] (enr1) {};
        \node[inner sep=2pt, fill=black,right=1cm of s0, yshift=0.8cm] (as0) {};
        \node[inner sep=2pt, fill=black,right=0.5cm of s0] (bs0) {};
         \node[inner sep=2pt, fill=black,right=1cm of s0, yshift=-0.8cm] (cs0) {};
        \node[interface,right=1.5cm of bs0,fill=white,yshift=0.7cm] (s3) {\scalebox{0.7}{$\exrarg{1}$}};
        \node[interface,right=-0.65cm of s3,fill=white, yshift=-1.2cm] (s4) {\scalebox{0.7}{$\exrarg{2}$}};
        \node[inner sep=0,right=0.4cm of s3] (exr1) {};
        \node[inner sep=0,right=0.4cm of s4] (exr2) {};
        \draw[->] (enr1) -> (s0);
        \draw[->] (s0) -> node [above] {$\scalebox{0.7}{$\point_1$}$} (as0);
        \draw[->] (s0) -> node [above, xshift=0.1cm] {$\scalebox{0.7}{$\point_2$}$} (bs0);
        \draw[->] (s0) -> node [below] {$\scalebox{0.7}{$\point_3$}$} (cs0);
        \draw[->] (as0) -> node [above] {$\scalebox{0.7}{0.3}$} (s3);
        \draw[->] (as0) -> node [above, yshift=0.1cm] {$\scalebox{0.7}{0.1}$} (s4);
        \draw[->] (bs0) -> node [below, yshift=0.05cm] {$\scalebox{0.7}{0.2}$} (s3);
        \draw[->] (bs0) -> node [above, yshift=-0.05cm] {$\scalebox{0.7}{0.4}$} (s4);
        \draw[->] (cs0) -> node [below, yshift=-0.15cm] {$\scalebox{0.7}{0.27}$} (s3);
        \draw[->] (cs0) -> node [below] {$\scalebox{0.7}{0.3}$} (s4);
        \draw[->] (s3) -> (exr1);
        \draw[->] (s4) -> (exr2);
    \end{tikzpicture}
    \subcaption{Shortcut oMDP $\cmdp{\fachievable{\mdp{A}}}$.}\label{subfig:ExShortcutMDP}
    \end{minipage}
    \caption{Two oMDPs. We omit transitions to a sink for readability. }
    \label{fig:exParetoOpenMDPs}
\end{figure} 

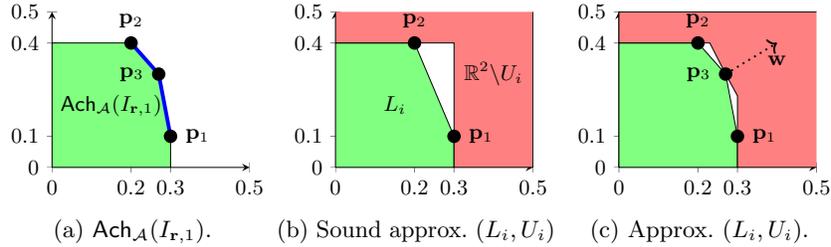
\begin{figure}[t]
    \centering
    \begin{minipage}[t]{0.3\hsize}
    \centering
    \pgfplotsset{width=4.5cm,compat=1.9}
			
 \scalebox{0.9}{
 \begin{tikzpicture}
        \begin{axis}[
            axis lines = left,
            xmin=0, xmax=0.5,
            ymin=0, ymax=0.5,
            xtick={0, 0.2, 0.3, 0.5},
            ytick={0, 0.1, 0.4, 0.5},
            axis background/.style={fill=white},
        ]
        \addplot[fill=green!50, very thin] coordinates {(0.3,0) (0.3,0.1) (0.27, 0.3) (0.2,0.4) (0,0.4) (0, 0)} -- cycle;
        \addplot[
            color=blue,
            ultra thick
        ]
        coordinates {
            (0.2,0.4) (0.27,0.3)
        };
        \addplot[
            color=blue,
            ultra thick
        ]
        coordinates {
            (0.27,0.3) (0.3, 0.1)
        };
        \node[point,label=0:$\point_1$] at (axis cs:0.3,0.1) {};
	\node[point,label=90:$\point_2$] at (axis cs:0.2,0.4) {};
 \node[point,label=180:$\point_3$] at (axis cs:0.27,0.3) {};	
        \node[] at (axis cs:0.15,0.2) {$\achievable{\mdp{A}}{\enrarg{1}}$};
        \end{axis}
        \end{tikzpicture}
    }
    \subcaption{$\achievable{\mdp{A}}{\enrarg{1}}$.}\label{subfig:ExPareto}
    \end{minipage}
    \begin{minipage}[t]{0.3\hsize}
    \centering
    \pgfplotsset{width=4.5cm,compat=1.9}
    \scalebox{0.9}{
        \begin{tikzpicture}
        \begin{axis}[
            axis lines = left,
            xmin=0, xmax=0.5,
            ymin=0, ymax=0.5,
            xtick={0, 0.2, 0.3, 0.5},
            ytick={0, 0.1, 0.4, 0.5},
            axis background/.style={fill=white},
        ]
        \addplot[fill=green!50, very thin] coordinates {(0.3,0) (0.3,0.1) (0.2,0.4) (0,0.4) (0, 0)} -- cycle;
        \addplot[fill=red!50, very thin] coordinates {(0.5, 0.5) (0.5, 0) (0.3,0) (0.3,0.4) (0,0.4) (0, 0.5)};

        \node[point,label=0:$\point_1$] at (axis cs:0.3,0.1) {};
	\node[point,label=90:$\point_2$] at (axis cs:0.2,0.4) {};
        \node[] at (axis cs:0.15,0.2) {$L_i$};
	\node[] at (axis cs:0.4,0.3) {$\real^{2}\backslash U_i$};		
        \end{axis}
        \end{tikzpicture}
        }
    \subcaption{Sound approx.\ $(L_i,U_i)$ }\label{subfig:AxisPareto}
    \end{minipage}
     \begin{minipage}[t]{0.3\hsize}
    \centering
    \pgfplotsset{width=4.5cm,compat=1.9}
        \scalebox{0.9}{
        \begin{tikzpicture}
        \begin{axis}[
            axis lines = left,
            xmin=0, xmax=0.5,
            ymin=0, ymax=0.5,
            xtick={0, 0.2, 0.3, 0.5},
            ytick={0, 0.1, 0.4, 0.5},
            axis background/.style={fill=white},
        ]
        \addplot[fill=green!50, very thin] coordinates {(0.3,0) (0.3,0.1) (0.27, 0.3) (0.2,0.4) (0,0.4) (0, 0)} -- cycle;
        \addplot[fill=red!50, very thin] coordinates {(0.5, 0.5) (0, 0.5) (0, 0.4)  (0.23,0.4) (0.27,0.3) (0.3, 0.23) (0.3, 0) (0.5, 0)};
        \addplot[
            color=black,
            dotted,
            ->,
            thick
        ]
        coordinates {
            (0.27,0.3)(0.4,0.4)
        };
        \node[point,label=0:$\point_1$] at (axis cs:0.3,0.1) {};
	\node[point,label=90:$\point_2$] at (axis cs:0.2,0.4) {};
 \node[point,label=180:$\point_3$] at (axis cs:0.27,0.3) {};	
 \node[label=90:$\convw{w}$] at (axis cs:0.4,0.28) {};
        \end{axis}
        \end{tikzpicture}
        }
    \subcaption{Approx.\ $(L_i, U_i)$.}\label{subfig:UpdatingPareto}
    \end{minipage}
    \caption{Pareto curves (blue) and achievable regions and their underapproximation (green) as well as their overapproximation (white+green). }
    \label{fig:ParetoOpenMDPs}
\end{figure}

\begin{example}
Consider the oMDP $\mdp{A}$ in~\cref{subfig:ExOpenMDP}. From the entrance, the scheduler $\sigma_1$ that always chooses the $a$ yields reachability probabilities $(0.2, 0.4)$ to reach the first and second exit, respectively, while $\sigma_2$ that chooses $b$ in $s_1$ yields $(0.3, 0.1)$. 
\end{example}

\newcommand{\vtor}[1]{\langle #1 \rangle}
\paragraph{Geometry.}
For two points $\point, \point' \in \mathbb{R}^n$, we write $\point \preceq \point'$ for pointwise inequality, i.e., $p_j \leq p'_j$ for every $j \leq n$.  
The set of (normalized) weight vectors of dimension $n \in \nat$ is given by $\weights{n} = \{ \convw{w} = \vtor{w_1, \dots, w_n} \in \mathbb{R}^n_{\geq 0} \mid \sum_{i} w_i = 1 \}$. 
The \emph{convex closure} of a set $X$ is the set $\convcl{X} \defeq \{ \sum w_i \cdot \point_i \mid w \in \weights{n},\, \point_1,\dots,\point_n \in X \}$. 
The \emph{downward closure} of a set $X$ is the set 
$\dw{X} \defeq \{ \point \mid \exists \point' \in X. \point \preceq \point'  \}$. The downward convex closure is $\dwconvcl{X} \defeq \dw{\convcl{X}}$. 
A set $X$ is convex or downward-closed, if $X$ is equal to its convex- or downward-closure, respectively. A convex, downward-closed set is \emph{finitely generated}, if it is the  downward convex closure of a finite set of points that we call \emph{vertices}. 

\smallskip 
\noindent
We write $\Reacha{\mdp{A}, \sigma}{i}{\ex} \in \mathbb{R}^\ex$ for a point with $\Reacha{\mdp{A}, \sigma}{i}{\ex}(j) = \Reacha{\mdp{A}, \sigma}{i}{j}$.

\begin{definition}[Achievable]
\label{def:achievable}
For an oMDP $\mdp{A}$, $i \in \en^\mdp{A}$ and a scheduler $\sigma$  
\[ \achievablesched{\mdp{A}}{i}{\sched} \defeq \{ \point \in \mathbb{R}^\ex \mid  \point \preceq \Reacha{\mdp{A},\sched}{i}{\ex} \} \] 
is the points achieved by $\sigma$. For a set of schedulers $\Sigma$, 
$
\achievablesched{\mdp{A}}{i}{\Sigma} \defeq \bigcup_{\sigma \in \Sigma} {\achievablesched{\mdp{A}}{i}{\sigma}}.  
$
The set of \emph{achievable points} is $\achievable{\mdp{A}}{i} \defeq \achievablesched{\mdp{A}}{i}{\Gmsched{\mdp{A}}}$. 
The \emph{Pareto curve} is the set \[ \semFunctor(\mdp{A}, i) \defeq \{ \point \in \achievable{\mdp{A}}{i} \mid \text{for all } \point' \succ \point: \point' \not\in \achievable{\mdp{A}}{i}    \}. \] A scheduler $\sigma$ is Pareto-optimal (w.r.t. $i$), if $\Reacha{\mdp{A}, \sigma}{i}{O} \in \semFunctor(\mdp{A}, i)$. 
\end{definition}
Points in the Pareto curve are also called \emph{Pareto-optimal}. The set of achievable points is convex and downward-closed, the downward closure of the Pareto curve is the set of achievable points~\cite{EtessamiKVY08,ForejtKNPQ11}. \cref{subfig:ExPareto} illustrates the achievable points (the green region)  and the Pareto curve (the blue line) of $\mdp{A}$ from the entrance $\enrarg{1}$. 
From facts on the necessary schedulers in multi-objective model-checking~\cite{EtessamiKVY08} with the fact that exits are absorbing, we only need to consider memoryless schedulers and the vertices of the achievable points all correspond to DM schedulers:
\begin{lemma}\label{lem:dmsuffices} For oMDP $\mdp{A}$ and entrance $i \in \en$: 
 \[ \dw{\semFunctor(\mdp{A}, i)} = \achievable{\mdp{A}}{i} = \achievablesched{\mdp{A}}{i}{\Possched{\mdp{A}}} = \dwconvcl{\achievablesched{\mdp{A}}{i}{\Dmsched{\mdp{A}}}}. \]
\end{lemma}
We define the \emph{error} $\gap{X}{Y}$ between downward-closed convex sets $X \subseteq Y \subseteq \mathbb{R}^n$ using the $L^2$ norm~$\|\_\|_2$ as in~\cite{DBLP:journals/fmsd/QuatmannJK22}:
$
    \gap{X}{Y} \defeq  \sup_{\point \in Y}  \inf_{\point' \in X}\|\point - \point'  \|_2
$.

To simplify notation, we write $\familyachievablesched{\mdp{A}}{\Sigma}$ for the indexed family $\big(\achievablesched{\mdp{A}}{i}{\Sigma}\big)_{i\in \en}$ of the achievable points.  Given families $X = \{X_i\}_{i \leq k}$ and $Y =\{Y_i\}_{i \leq k}$ with $X(i) \subseteq Y(i)$, we  define $\gap{X}{Y} \defeq \sup_{i \leq k} \gap{X_i}{Y_i}$.  

We conservatively extend the single-exit problem statement. We now want to find a \emph{set} of schedulers such 
that the error between the achieved points corresponding to this set and the (unknown) Pareto curve is small.
\begin{mdframed}
{\bf Multi-Exit Problem Statement:} 
Given a string diagram $\sd{D}$, an error bound $\epsilon \in \probinterval$,
compute a set of schedulers $\Sigma$ s.t.\ $\gap{\familyachievablesched{\mdp{A}}{\Sigma}} {\fachievable{\mdp{A}}} \leq \epsilon$.
\end{mdframed}

\section{Compositional Algorithm}
\label{sec:naive_comp_algo}

We present a compositional algorithm by soundly approximating Pareto curves. 



\subsection{Approximating Pareto Curves on oMDPs}

For an oMDP $\mdp{A}$, we can efficiently approximate $\fachievable{\mdp{A}}$ via an off-the-shelf multi-objective model checking algorithm~\cite{ForejtKP12,QuatmannK21}. We  outline this algorithm, tailored to oMDPs. The key idea is that the algorithm iteratively refines a sound approximation of the Pareto curve. 

\begin{definition}
\label{def:underoverapprox}
    Let $\mdp{A}$ be an oMDP. An \emph{under-approximation} $L$ is a family $L \defeq (L_i)_{i\in \en}$ such that $L_i$ is a convex, downward-closed, and finitely generated with $L_i\subseteq \achievable{\mdp{A}}{i}$ for  $i\in \en$. 
    An \emph{over-approximation} $U$ is analogously defined, with 
    $\achievable{\mdp{A}}{i} \subseteq U_i$ for $i\in \en$. 
    We call $(L,U)$ a \emph{sound approximation} for $\mdp{A}$.
\end{definition}

\begin{algorithm}[t]
\caption{\texttt{approxMultiObjMDP}: Approximation of the Pareto curve}
\begin{algorithmic}[1]
\State \textbf{Input:} an oMDP $\mdp{A}$ and an imprecision $\eta\in \probinterval$.
\State \textbf{Output:} a sound approximation $(L,U)$ of $\mdp{A}$ s.t.\ $\gap{L}{U}\leq \eta$.
\State initialize $L_i \defeq \emptyset$ and $U_i\defeq \sdist{\ex}$, for $i\in \en$.  
\While{$\gap{L}{U} > \eta$}
\State select an entrance $i\in \en$, a weight vector $\convw{w}\defeq (w_o)_{o\in \ex}$, and $\delta \geq 0$. \label{select_weights_error}
\State find  $l_{\convw{w}}, u_{\convw{w}}$ s.t.\   $l_{\convw{w}} \leq   \sup_{\sigma} \WReacha{\sigma}{\convw{w}}{i} \leq u_{\convw{w}}$,  $|l_{\convw{w}}-u_{\convw{w}}|\leq \delta$.
\State compute $\point^{l}$ such that $l_{\convw{w}} = \convw{w} \cdot \point^{l}$. 
\State update $L_i  \defeq L_i \cup \{\point^{l}\}$ and $U_i \defeq U_i \cap  \{\point \mid \convw{w}\cdot \point \leq u_{\convw{w}}\}$.
\EndWhile
\State \textbf{return} $(L, U)$.
\end{algorithmic}
\label{alg:approxMultiObjMDP}
\end{algorithm}
\cref{alg:approxMultiObjMDP} summarizes the approach. 
$L$ and $U$ are initialized as a trivial approximation. The algorithm iteratively refines them by computing \emph{weighted reachability} probabilities for some weight vector $\convw{w}  \in \weights{\ex}$, which is adequately chosen~\cite{QuatmannK21}\footnote{While our algorithm is indeed correct for $\delta > 0$, we only discuss $\delta=0$ here.}. For that $\convw{w}$, we denote the weighted reachability using $\WReacha{\mdp{A},\sched}{\convw{w}}{i} \defeq \sum_{j \in \ex} \convw{w}_j \cdot \Reacha{\mdp{A},\sched}{i}{j}$. A scheduler $\tau$ is \emph{weighted optimal} w.r.t.\  $\convw{w}$ and $i$ if  $\WReacha{\mdp{A},\tau}{\convw{w}}{i} =\sup_{\sigma\in \Scheds{\mdp{A}}} \WReacha{\mdp{A},\sigma}{\convw{w}}{i}$.
Weighted reachability can be computed via standard reachability query on a mildly modified MDP that has a fresh unique target and sink state. This implies that DM schedulers suffice for optimality. 
Furthermore, every Pareto optimal scheduler optimizes some weighted reachability. Finally, from a lower bound  $l_{\convw{w}}$ we can compute an achievable $\point^l$ and use this point to update $L$.
As the weighted optimal scheduler is optimal, we obtain a linear inequality on the reachability probabilities that can be achieved in the direction of $\convw{w}$\footnote{By construction, $U_i$ is convex, downward-closed and finitely generated.}. \Cref{subfig:UpdatingPareto} illustrates how~\cref{alg:approxMultiObjMDP} works. 

\begin{proposition}[\!\!\cite{QuatmannK21}]
\label{prop:algmdpcorrect}
    \Cref{alg:approxMultiObjMDP} is correct.
\end{proposition}


\subsection{From Pareto Curves to Shortcut MDPs}
\label{sec:paretotoshortcut}
In our algorithm, it will be convenient to construct a (small) MDP with a particular Pareto curve. In particular, we construct the oMDP in \cref{subfig:ExShortcutMDP} from the Pareto curve in \cref{subfig:ExPareto}.
This construction is rather straightforward and we give it for both finite sets of (Pareto-optimal) points and for finitely generated convex sets.
Due to the exits being terminals, $\Reacha{\mdp{A}, \sigma}{i}{\ex} \in \sdist{\ex}$, where we liberally interpret distributions as points in~$\mathbb{R}^O$.


\begin{definition}[shortcut oMDP]
\label{def:abstractMDP}Let $\mdp{A}$ be an oMDP, $B$ be an indexed family $B\defeq (B_i)_{i\in \en^\mdp{A}}$ of finite sets $B_i \subseteq \sdist{\ex^\mdp{A}}$ 
The \emph{shortcut oMDP for $B$} is  $\cmdp{B} \defeq (M, \interface^{\mdp{A}})$, with $M\defeq (S,A, P)$, $S\defeq \en^\mdp{A} \cup \ex^\mdp{A} \cup \{\star\}$, 
$ A \defeq  \bigcup_{i\in \en^\mdp{A}} {B_i}$,  
\begin{align*}
    P(s, a, s') &\defeq \begin{cases}
        a(s') &\text{ if $s\in \en^\mdp{A}$, $a\in B_s$, and $s'\in \ex^\mdp{A}$,}\\
        1 - \sum_{o\in \ex^\mdp{A}} a(o) &\text{ if $s\in \en^\mdp{A}$, $a\in B_s$, and $s' = \star$,}\\
        0 &\text{ otherwise. }
    \end{cases}
\end{align*}
Additionally, if $B$ is an indexed family $B\defeq (B_i)_{i\in \en}$  of convex, downward-closed and finitely generated $B_i \subseteq \sdist{O}$ with vertices $B^V_i$, then $\cmdp{B} \defeq \cmdp{(B_i^V)_{i \in \en}}$.
\end{definition}

 For each entrance $s$, exit $s'$ and point $\point \in B_s$, the probability transition $P(s, \point, s')$ is the reachability probability from $s$ to $s'$ in $\mdp{A}$ induced by schedulers that yielded reachability probabilities in point $\point$. The sink $\star$ is introduced to ensure that we obtain proper distributions.
\begin{proposition}
\label{prop:correct_abst_openMDPs}
    For oMDP $\mdp{A}$, $i \in \en^\mdp{A}$, and sound approximation $(L, U)$: 
    \begin{align*}
  \achievable{\cmdp{L}}{i} \quad \subseteq \quad  \achievable{\mdp{A}}{i} \quad \subseteq\quad  \achievable{\cmdp{U}} {i}.
    \end{align*}
\end{proposition}
   For the inclusions, a key element in the correctness of this statement that, intuitively, we can add points that are not vertices of $B$ without changing the Pareto curve $\cmdp{B}$. 
\iffull
\cref{lem:addingSuboptimalActions} in \cref{app:proof:addingSuboptimalActions} formalizes this idea for any MDP. 
Using that lemma, we establish the following implication for $X, Y \subseteq \dist{O}^\en$ where $X(i), Y(i)$ are finitely generated downward-closed convex sets:
$
X(i) \subseteq Y(i)\text{ implies }\achievable{\cmdp{X}}{i} \subseteq \achievable{\cmdp{Y}}{i}
$.
Finally, we apply this implication twice. Details are given in \cref{app:proof:correct_abst_openMDP}.
\else Based on the idea, we establish the following implication for $X, Y \subseteq \dist{O}^\en$ where $X(i), Y(i)$ are finitely generated downward-closed convex sets:
$
X(i) \subseteq Y(i)\text{ implies }\achievable{\cmdp{X}}{i} \subseteq \achievable{\cmdp{Y}}{i}
$.
Finally, we apply this implication twice. Details are given in~\cite[Appendix B]{WVHRJ2024accepted}.
\fi

~\cref{prop:correct_abst_openMDPs} also implies the following corollary, which claims that Pareto-optimal schedulers suffice for optimality: 
\begin{corollary}
\label{cor:pareto_optimal}
For oMDP $\mdp{A}$, 
\begin{align*}
    \fachievable{\mdp{A}} =  
   \fachievable{\cmdp{\fachievable{\mdp{A}}}}.
\end{align*}
\end{corollary}

It is often convenient to obtain schedulers on either $\cmdp{\fachievable{\mdp{A}}}$ and obtain a scheduler on $\mdp{A}$ or vice versa. Furthermore, while \cref{prop:correct_abst_openMDPs} considers every entrance separately, we aim to have schedulers that match on \emph{every} entrance \emph{simultaneously} by remembering which entrance was taken (\iffull proof in \cref{app:proof:construct_deterministic_scheduler}\else proof in \cite[Appendix B]{WVHRJ2024accepted}\fi).
\begin{proposition}
\label{lem:construct_deterministic_scheduler}
\label{lem:recoverstratfromshortcut}
Let $\mdp{A}$ be an oMDP and $\sigma \in \Dmsched{\cmdp{\fachievable{\mdp{A}}}}$. 
There is a scheduler $\tau$ on $\mdp{A}$ s.t.\
$\Reacha{\mdp{A},\tau}{i}{\ex^{\mdp{A}}} = \Reacha{\cmdp{\fachievable{\mdp{A}}},\sigma}{i}{\ex^{\mdp{A}}}$ for every $i\in \en^{\mdp{A}}$. Let $\tau' \in \Dmsched{\mdp{A}}$, there is $\sigma' \in \Dmsched{\cmdp{\fachievable{\mdp{A}}}}$ s.t.\ $\Reacha{\mdp{A},\tau'}{i}{\ex^{\mdp{A}}} \preceq \Reacha{\cmdp{\fachievable{\mdp{A}}},\sigma'}{i}{\ex^{\mdp{A}}}$ for every $i\in \en^{\mdp{A}}$.
\end{proposition}

\begin{algorithm}[t]
\caption{\texttt{approxMultiObjSD}: Approximation of the Pareto curve}\label{alg:AppMeager}
\begin{algorithmic}[1]
\State \textbf{Input:} a string diagram $\sd{D}$ and a local imprecision $\eta\in \probinterval$.
\State \textbf{Output:} A sound approximation $(L, U)$ for $\semantics{\sd{D}}$. If $\eta = 0$, then $L = U$.
\If{$\sd{D} = \mdp{A}$} 
\State \Return \texttt{approxMultiObjMDP}$(\mdp{A}, \eta)$ \Comment{Invoke Algorithm~\ref{alg:approxMultiObjMDP}}
\Else \Comment{$\sd{D} = \sd{D}_1 \ast \sd{D}_2$, $\ast \in \{ \seqcomp, \oplus \} $}
 \State $(L_1, U_1) \gets \texttt{approxMultiObjSD}(\sd{D}_1, \eta)$ \Comment{Recursion}
    \State $(L_2, U_2) \gets \texttt{approxMultiObjSD}(\sd{D}_2, \eta)$  \Comment{Recursion}
    \State $\mdp{A}_L \gets 
  \cmdp{L_1} \ast  \cmdp{L_2}$ \Comment{See Def.~\ref{def:abstractMDP}}
    \State $\mdp{A}_U \gets 
  \cmdp{U_1} \ast  \cmdp{U_2}$ \Comment{See Def.~\ref{def:abstractMDP}}
    \State \Return $\Big(\texttt{approxMultiObjMDP}(\mdp{A}_L, \eta)\Big)_1, \Big(\texttt{approxMultiObjMDP}(\mdp{A}_U, \eta)\Big)_2$
\EndIf
\end{algorithmic}
\label{alg:approxMultiObjSD}
\end{algorithm}

\subsection{Approximating Pareto Curves for String Diagrams}
\label{sec:composingshortcuts}
We now provide a recursive algorithm in Algorithm~\ref{alg:approxMultiObjSD}. In the base case, when is a single oMDP, we analyze the oMDP using Algorithm~\ref{alg:approxMultiObjMDP}.
Otherwise, we have an string diagram $\sd{D}_1 \ast \sd{D}_2$ for $\ast \in \{ \seqcomp, \oplus \}$. We recursively compute sound approximations for $\sd{D}_1$ and  $\sd{D}_2$.
Next, we compose the under- and over-approximations, respectively. We discuss the under-approximation, the over-approximation is handled analogously. Given both under-approximations $L_1, L_2$, we create the corresponding shortcut MDPs $\cmdp{L_1}$ and $\cmdp{L_2}$ and then take their sequential composition $\mdp{A}_L$. This oMDP can be analyzed using Algorithm~\ref{alg:approxMultiObjMDP}. Any under-approximation for $\mdp{A}_L$ is an underapproximation for $\sd{D}(\defeq \sd{D}_1 \ast \sd{D}_2)$. 
This algorithm  easily  supports additional operations such as $n$-ary compositions that significantly reduce overhead.
We remark that contrary to \cref{alg:approxMultiObjMDP}, we do not guarantee any error on the approximation that we return, unless all computations are precise ($\eta = 0$). We discuss error bounds in \cref{sec:comp_error_bouds}.

\subsubsection{Correctness}
We first state that under-approximations and over-approximations are preserved by the algebraic operations:

\begin{proposition}[Case $\seqcomp$]
\label{prop:correct_seqcomp}
    For oMDPs $\mdp{A},\mdp{B}$, and sound approximations $(L^\mdp{A}, U^\mdp{A})$ and $(L^\mdp{B}, U^\mdp{B})$, the tuple $\Big(\fachievable{\cmdp{L^\mdp{A}}\seqcomp\cmdp{L^\mdp{B}}},\fachievable{\cmdp{U^\mdp{A}}\seqcomp\cmdp{U^\mdp{B}}} \Big)$ is a sound approximation for $\mdp{A}\seqcomp \mdp{B}$, i.e., for any $i \in \en^\mdp{A\seqcomp B}$, the following conditions hold:
    \begin{align*}  
&\achievable{\cmdp{L^\mdp{A}}\seqcomp\cmdp{L^\mdp{B}}}{i} \quad \subseteq \quad 
  \achievable{\mdp{A}\seqcomp\mdp{B}}{i}  \quad \subseteq\quad  \achievable{\cmdp{U^\mdp{A}}\seqcomp\cmdp{U^\mdp{B}}}{i}.
    \end{align*}
\end{proposition}
 \noindent The proof (\cref{sec:seqcompcorrect}) is not trivial due to the cyclic dependency between $\mdp{A}$ and $\mdp{B}$.
Similar to~\cref{cor:seqcom_pareto_optimal}, Pareto-optimal schedulers suffice for compositionally solving Pareto curves w.r.t. the sequential composition. 
\begin{corollary}
\label{cor:seqcom_pareto_optimal}
    For oMDPs $\mdp{A},\mdp{B}$,
    \begin{align*}
        \fachievable{\mdp{A}\seqcomp \mdp{B}} =  
   \fachievable{\cmdp{\fachievable{\mdp{A}}}\seqcomp \cmdp{\fachievable{\mdp{B}}}}. 
    \end{align*}
\end{corollary}

The similar compositionality  result also holds for the sum: 
\begin{proposition}[Case $\oplus$]
\label{prop:correct_oplus}
    For oMDPs $\mdp{A},\mdp{B}$, and sound approximations $(L^\mdp{A}, U^\mdp{A})$ and $(L^\mdp{B}, U^\mdp{B})$, the tuple $\Big(\fachievable{\cmdp{L^\mdp{A}}\oplus\cmdp{L^\mdp{B}}},\fachievable{\cmdp{U^\mdp{A}}\oplus\cmdp{U^\mdp{B}}} \Big)$ is a sound approximation for $\mdp{A}\oplus\mdp{B}$, i.e., for any $i \in \en^\mdp{A\oplus B}$, the following conditions hold:
    \begin{align*}   
&\achievable{\cmdp{L^\mdp{A}}\oplus\cmdp{L^\mdp{B}}}{i} \quad \subseteq \quad 
    \achievable{\mdp{A}\oplus\mdp{B}}{i}  \quad \subseteq\quad  \achievable{\cmdp{U^\mdp{A}}\oplus\cmdp{U^\mdp{B}}}{i}.
    \end{align*}
\end{proposition}
The statement is straightforward due to the lack of interaction between $\mdp{A}$ and $\mdp{B}$. 
\iffull 
\begin{corollary}
\label{cor:sum_pareto_optimal}
    For oMDPs $\mdp{A},\mdp{B}$,
    \begin{align*}
        \fachievable{\mdp{A}\oplus \mdp{B}} =  
   \fachievable{\cmdp{\fachievable{\mdp{A}}}\oplus \cmdp{\fachievable{\mdp{B}}}}. 
    \end{align*}
\end{corollary}
\else 
\fi

\begin{theorem}
Algorithm~\ref{alg:approxMultiObjSD} is correct. Additionally, if $\eta=0$ then $\gap{L}{U} = 0$.
\end{theorem}
\begin{proof}
We prove that output $(L, U)$ is a sound approximation of $\semantics{\sd{D}}$. We prove this recursively over the structure of $\sd{D}$.  If $\sd{D} = \mdp{A}$, \cref{prop:algmdpcorrect} applies. Otherwise, we focus on $L$ being a lower bound and with $\ast = \seqcomp$ the upper bound and $\oplus$ are analogous. Indeed $L(i) \subseteq \achievable{\cmdp{L_1}\seqcomp\cmdp{L_2}}{i}$ for any entry $i$ by Algorithm~\ref{alg:approxMultiObjMDP}. Likewise, $L_1$, $L_2$ are lower bounds to $\semantics{\sd{D}_1}$ and $\semantics{\sd{D}_2}$ and thus the theorem follows \cref{prop:correct_seqcomp} and \cref{prop:correct_oplus}. If $\eta=0$, then $L$ and $U$ recursively coincide. \qed
\end{proof}

\subsubsection{Obtaining schedulers}
To obtain a scheduler, first observe that in Algorithm~\ref{alg:approxMultiObjMDP}, every point in $L_i$ can be annotated with a memoryless scheduler (using standard model checking).
Now, when we obtain a memoryless scheduler for some $\mdp{A} \ast \mdp{B}$, then we can translate this straightforwardly to memoryless schedulers for $\mdp{A}$ and $\mdp{B}$. Finally, if we obtain a scheduler for $\cmdp{L'}$ and $L'$ is an underapproximation for $\mdp{A}$, then \cref{lem:recoverstratfromshortcut} states that we can recover a scheduler for $\mdp{A}$.

\subsection{Proof outline for \cref{prop:correct_seqcomp}}

\label{sec:seqcompcorrect}
We give the main ingredients for \cref{prop:correct_seqcomp}, the key ingredient for our approach. 
We discuss the crux for showing   
 $ \achievable{\cmdp{\fachievable{\mdp{A}}}\seqcomp\cmdp{\fachievable{\mdp{B}}}}{i}  =
    \achievable{\mdp{A}\seqcomp\mdp{B}}{i}$:
We can map (memoryless) schedulers in $\mdp{A}\seqcomp\mdp{B}$ and $\cmdp{\fachievable{\mdp{A}}}\seqcomp\cmdp{\fachievable{\mdp{B}}}$ to each other while matching reachability probabilities. More precisely, we lift \cref{lem:recoverstratfromshortcut} to the sequential composition, while using that \cref{lem:recoverstratfromshortcut} already established the mapping for $\mdp{A}$ and $\cmdp{\fachievable{\mdp{A}}}$.  
Therefore, we note that for any $\sched$, $i\in \en^{\mdp{A}}\cup  \en^{\mdp{B}}$ and $j\in \ex^{\mdp{A}\seqcomp \mdp{B}}$, the following equations hold directly from the definition of the sequential composition and from the definition of reachability probabilities by adequately partitioning the paths from entrance to exit, i.e. $\Reacha{\mdp{A}\seqcomp \mdp{B}, \tau}{i}{j} =$
    \begin{align*}
        \begin{cases}\Reacha{\mdp{A},\tau}{i}{j} + \sum_{k\in \exr^{\mdp{A}}} \Reacha{\mdp{A},\tau}{i}{k}\cdot  \Reacha{\mdp{A}\seqcomp \mdp{B}, \tau}{\nxt(k)}{j} &\text{ if }i\in \en^{\mdp{A}},\\
       \Reacha{\mdp{B},\tau}{i}{j} + \sum_{k\in \exl^{\mdp{B}}} \Reacha{\mdp{B},\tau}{i}{k}\cdot  \Reacha{\mdp{A}\seqcomp \mdp{B}, \tau}{\nxt(k)}{j} &\text{ if } i\in \en^{\mdp{B}},
        \end{cases}
    \end{align*}
 where $\nxt((\exr^\mdp{A})_i) = {(\enl^\mdp{B})}_i$ and  $\nxt((\exl^\mdp{B})_i) = {(\enr^\mdp{A})}_i$, are the next states that are visited from any exits which are \emph{not} in $\ex^{\mdp{A}\seqcomp \mdp{B}}$.
 Naturally, by substitution we obtain the equations for $\cmdp{\fachievable{\mdp{A}}} \seqcomp \cmdp{\fachievable{\mdp{B}}}$. 
 We observe that various parts of these equations are independent of the sequential composition. Thus, for these, \cref{lem:recoverstratfromshortcut} applies. Once we apply these, we obtain two times the \emph{same} linear equation system with variables for 
 $\Reacha{\mdp{A}\seqcomp \mdp{B}, \tau}{i}{j}$ and $\Reacha{\cmdp{\fachievable{\mdp{A}}}\seqcomp \cmdp{\fachievable{\mdp{B}}}, \tau}{i}{j}$, respectively, which shows that the probabilities coincide.
\iffull In \cref{app:explainmainprop}\else In~\cite[Appendix C]{WVHRJ2024accepted}\fi, we derive the inclusions formally and show that it indeed preserves reachability probabilities. We also establish the inclusions in \cref{prop:correct_seqcomp} analogously to \cref{prop:correct_abst_openMDPs}.


\section{Compositional Estimation of Error Bounds}
\label{sec:comp_error_bouds}

As we discussed above, \cref{alg:approxMultiObjSD} provides a way to obtain the Pareto curve \emph{precisely}, for $\eta=0$. However, setting the imprecision $\eta = 0$ is often infeasible.  When setting an $\eta > 0$, we only have soundness of the bounds, but no guarantee on the tightness. A naive extension to \cref{alg:approxMultiObjSD} would be to \emph{a posteriori} determine the error and tighten the sound approximation if the error is not matched. This process terminates with the required error bound as there are only finitely many schedulers. In this section, we discuss the (im)possibility of a one-shot approach, where we would recursively compute sound approximations with an approximate error bound given \emph{a priori}. 
To that end, we study how the error propagates through the composition. We show positive results by restricting the compositional structure: we maintain errors on string diagrams that are only constructed by the sum and the rightward sequential composition, as we show in~\cref{prop:error_sum,thm:error_uni_seq}. After showing negative results that explode the error on the (general) sequential composition in~\cref{ex:err_bounds_seq_comp_steq}, we end with a positive note, showing that the final result can have an error that is (significantly) \emph{smaller} than the individual errors in~\cref{ex:err_bounds_disappear}.

\paragraph{The $L^{\infty}$-Error}
For conciseness, this section uses the $L^{\infty}$-error between sound approximations. 
Let $(L, U)$ be a sound approximation. The \emph{$L^{\infty}$-error} is $\gapinf{L}{U}\defeq \sup_{i\in \en}\sup_{\point \in U_i} \inf_{\point'\in L_i} \|\point - \point'\|_{\infty}$, where $\|\_\|_{\infty}$  is the $L^{\infty}$ norm. 
The $L^{\infty}$-error $\gapinf{L}{U}$ and the error $\gap{L}{U}$ are equivalent in the sense that a sequence $\big(\gapinf{L_n}{U_n}\big)_{n\in \nat}$ converges to $0$ iff  $\big(\gap{L_n}{U_n}\big)_{n\in \nat}$ converges to $0$. 


\paragraph{The sum.} For the sum $\mdp{A}\oplus \mdp{B}$, we can easily obtain an error bound compositionally, since there are no interactions between $\mdp{A}$ and $\mdp{B}$.  
\begin{proposition}
\label{prop:error_sum}
Let $\mdp{A}, \mdp{B}$ be oMDPs, and $(L^{\mdp{A}}, U^{\mdp{A}}), (L^{\mdp{B}}, U^{\mdp{B}})$ be sound approximations. Then: 
$
    \gapinf{L^{\mdp{A}\oplus \mdp{B}}}{U^{\mdp{A}\oplus \mdp{B}}} \leq \max\big( \gapinf{L^{\mdp{A}}}{U^{\mdp{A}}} ,\gapinf{L^{\mdp{B}}}{U^{\mdp{B}}}\big), 
$
where $L^{\mdp{A}\oplus \mdp{B}} \defeq  \fachievable{\cmdp{L^{\mdp{A}}}\oplus \cmdp{L^{\mdp{B}}}}$, and $U^{\mdp{A}\oplus \mdp{B}} \defeq \fachievable{\cmdp{U^{\mdp{A}}}\oplus \cmdp{U^{\mdp{B}}}}$.

\end{proposition}

\paragraph{Rightward composition.}
 An open MPD $\mdp{A}$ is \emph{rightward} if
 $\typemdp{\mdp{A}} = (m, 0)\rightarrow (l, 0)$.
\begin{example}
Let $\mdp{A}, \mdp{B}$ be rightward oMDPs with $\typemdp{\mdp{A}} = (1, 0) \rightarrow (2, 0)$ and $\typemdp{\mdp{B}} = (2, 0) \rightarrow (1, 0)$, and $(L^{\mdp{A}}, U^{\mdp{A}}), (L^{\mdp{B}}, U^{\mdp{B}})$ be sound approximations all generated by a singleton, such that we can write: 
\begin{align*}
    L^{\mdp{A}}_1 &\defeq (0.3, 0.2) ,  U^{\mdp{A}}_1 \defeq (0.4, 0.3),
 L^{\mdp{B}}_1 \defeq 0.7, L^{\mdp{B}}_2 \defeq 0.6, U^{\mdp{B}}_1\defeq 0.75, U^{\mdp{B}}_2 \defeq 0.65.
\end{align*}
Then the lower bound $L^{\mdp{A}\seqcomp \mdp{B}}$ on their composition  consists of one point $0.3 \cdot 0.7 + 0.2 \cdot 0.6 = 0.33$, and the upper bound consists of one point $0.4 \cdot 0.75 + 0.3 \cdot 0.65 = 0.495$. These values can be easily calculated from the shortcut MDPs. While the error was $0.1$ and $0.05$ respectively on $\mdp{A}$ and $\mdp{B}$, the composition has an error of $0.165$. 

\end{example}

We can estimate sufficient error bounds for rightward $\mdp{A}, \mdp{B}$ in order to ensure a certain error bound for the sequential composition $\mdp{A}\seqcomp \mdp{B}$: 
\begin{theorem}
\label{thm:error_uni_seq}
Let $\mdp{A}, \mdp{B}$ be rightward oMDPs, $(L^{\mdp{A}}, U^{\mdp{A}}), (L^{\mdp{B}}, U^{\mdp{B}})$ sound approximations, and $(L^{\mdp{A}\seqcomp \mdp{B}}, U_X), (L_X, U^{\mdp{A}\seqcomp \mdp{B}})$ be sound approximations of $\cmdp{L^{\mdp{A}}}\seqcomp \cmdp{L^{\mdp{B}}}$ and $\cmdp{U^{\mdp{A}}}\seqcomp \cmdp{U^{\mdp{B}}}$, respectively.
Then $\gapinf{L^{ \mdp{A}\seqcomp \mdp{B}}}{U^{ \mdp{A}\seqcomp \mdp{B}}}$ is bounded from above by 
\begin{equation*}
\label{eq:error_uniseq}
     |\ex^{\mdp{A}}|\cdot \gapinf{L^{\mdp{A}}}{U^{\mdp{A}}} + \gapinf{L^{\mdp{B}}}{U^{\mdp{B}}} + \gapinf{L^{\mdp{A}\seqcomp \mdp{B}}}{U_X} +  \gapinf{L_X}{U^{\mdp{A}\seqcomp \mdp{B}}}
\end{equation*}
\end{theorem}
\iffull See~\cref{subsec:proof_error_uni_seq} for the proof\else See~\cite[Appendix D]{WVHRJ2024accepted} for the proof\fi.
Whereas the first two summands are inherent to the approximation of $\mdp{A}$ and $\mdp{B}$, the latter two terms originate from the approximations when computing a Pareto curve of the composed shortcut MDPs.
 \cref{thm:error_uni_seq} thus provides reasonably tight error bounds for the sequential composition on rightward oMDPs (only) when the number of exits $\ex^{\mdp{A}}$ is small.

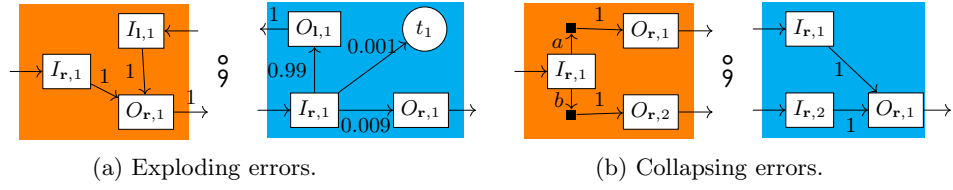
\begin{figure}[t]
\centering
\begin{minipage}{0.45\textwidth}
    \scalebox{0.9}{
\begin{tikzpicture}[
innode/.style={draw, rectangle, minimum size=0.5cm},
interface/.style={draw, rectangle, minimum size=0.5cm},]
\fill[orange] (-0.7cm, -1cm)--(-0.7cm, 1cm)--(1.8cm, 1cm)--(1.8cm, -1cm)--cycle;
\node[interface,fill=white] (s0) {$\enrarg{1}$};
\node[inner sep=0,right=-1.25cm of s0] (enr1) {};


\node[interface, right=0.4cm of s0, yshift = 0.6cm, fill=white] (s2) {$\enlarg{1}$};
\node[interface, right=0.4cm of s0, yshift = -0.6cm,fill=white] (s3) {$\exrarg{1}$};
\node[inner sep=0,right= 0.5cm of s2] (exr1) {};
\node[inner sep=0,right= 0.5cm of s3] (exr2) {};
\draw[->] (enr1) -> (s0);
\draw[->] (s0) -> node [above] {$1$} (s3);
\draw[->] (s2) -> node [left] {$1$} (s3);
\draw[->] (exr1) -> (s2);
\draw[->] (s3) -> node [above] {$1$} (exr2);
\node[right=0.5cm of s3, yshift= 0.6cm] (seqcomp) {\scalebox{2}{$\seqcomp$}};
\end{tikzpicture}
\hspace{1pt}
\begin{tikzpicture}[
innode/.style={draw, rectangle, minimum size=0.5cm},
interface/.style={draw, rectangle, minimum size=0.5cm},]
\fill[cyan] (-0.7cm, -1cm)--(-0.7cm, 1cm)--(2.2cm, 1cm)--(2.2cm, -1cm)--cycle;
\node[interface, fill=white, yshift=0.6cm] (t0) {$\exlarg{1}$};
\node[interface, fill=white, yshift=-0.6cm] (t1) {$\enrarg{1}$};
\node[inner sep=0, right=-1.25cm of t0] (enr1) {};
\node[inner sep=0, right=-1.25cm of t1] (enr2) {};
\node[interface, right=0.8cm of t1,fill=white] (t3) {$\exrarg{1}$};
\node[inner sep=0,right=0.4cm of t3] (exr1) {};
\node[state,right=0.9cm of t0, minimum size=0.5cm,fill=white] (t2) {$t_1$};

\draw[<-] (enr1) -> node [above] {$1$} (t0);
\draw[<-] (t1) -> (enr2);
\draw[->] (t1) -> node [below] {$0.009$} (t3);
\draw[->] (t1) -> node [above, yshift=0.1cm] {$0.001$} (t2);
\draw[<-] (t0) -> node [left] {$0.99$} (t1);
\draw[->] (t3) -> (exr1);
\end{tikzpicture}}\subcaption{Exploding errors.}
\label{fig:bidirectional_seqcomp}

\end{minipage}
\hfill
\begin{minipage}{0.45\textwidth}
    \scalebox{0.9}{
\begin{tikzpicture}[
innode/.style={draw, rectangle, minimum size=0.5cm},
interface/.style={draw, rectangle, minimum size=0.5cm},]
\fill[orange] (-0.7cm, -1cm)--(-0.7cm, 1cm)--(1.8cm, 1cm)--(1.8cm, -1cm)--cycle;
\node[interface,fill=white] (s0) {$\enrarg{1}$};
\node[inner sep=0,right=-1.25cm of s0] (enr1) {};


\node[interface, right=0.4cm of s0, yshift = 0.6cm, fill=white] (s2) {$\exrarg{1}$};
\node[inner sep=2pt, fill=black,above=0.3cm of s0] (a1a) {};

\node[inner sep=2pt, fill=black,below=0.3cm of s0] (a1b) {};

\node[interface, right=0.4cm of s0, yshift = -0.6cm,fill=white] (s3) {$\exrarg{2}$};
\node[inner sep=0,right= 0.5cm of s2] (exr1) {};
\node[inner sep=0,right= 0.5cm of s3] (exr2) {};
\draw[->] (enr1) -> (s0);
\draw[->] (s0) -> node [left] {$a$} (a1a);
\draw[->] (s0) -> node [left] {$b$} (a1b);
\draw[->] (a1a) -> node [above] {$1$} (s2);
\draw[->] (a1b) -> node [above] {$1$} (s3);

\draw[->] (s2) -> (exr1);
\draw[->] (s3) -> node [above] {} (exr2);
\node[right=0.5cm of s3, yshift= 0.6cm] (seqcomp) {\scalebox{2}{$\seqcomp$}};
\end{tikzpicture}
\hspace{-3pt}
\begin{tikzpicture}[
innode/.style={draw, rectangle, minimum size=0.5cm},
interface/.style={draw, rectangle, minimum size=0.5cm},]
\fill[cyan] (-0.7cm, -1cm)--(-0.7cm, 1cm)--(1.7cm, 1cm)--(1.7cm, -1cm)--cycle;
\node[interface, fill=white, yshift=0.6cm] (t0) {$\enrarg{1}$};
\node[interface, fill=white, yshift=-0.6cm] (t1) {$\enrarg{2}$};
\node[inner sep=0, right=-1.25cm of t0] (enr1) {};
\node[inner sep=0, right=-1.25cm of t1] (enr2) {};
\node[interface, right=0.5cm of t1,fill=white] (t3) {$\exrarg{1}$};
\node[inner sep=0,right=0.4cm of t3] (exr1) {};


\draw[->] (enr1) -> node [above] {} (t0);
\draw[<-] (t1) -> (enr2);
\draw[->] (t1) -> node [below] {$1$} (t3);
\draw[->] (t0) -> node [left] {$1$} (t3);
\draw[->] (t3) -> (exr1);
\end{tikzpicture}
}
\subcaption{Collapsing errors.}
\label{fig:bidirectional_seqcomp_goodcase}

\end{minipage}

\caption{Bidirectional sequential compositions $\mdp{A}\seqcomp \mdp{B}$.  }
\end{figure}

\paragraph{(General, bidirectional) sequential composition.}
In general, we cannot obtain tight error bounds for $\mdp{A}\seqcomp \mdp{B}$, even if their error bounds of $\mdp{A}$ and $ \mdp{B}$ are small. 
\begin{example}
\label{ex:err_bounds_seq_comp_steq}
    For oMDPs $\mdp{A}, \mdp{B}$ 
     in~\cref{fig:bidirectional_seqcomp}, and $(L^{\mdp{A}}, U^{\mdp{A}}), (L^{\mdp{B}}, U^{\mdp{B}})$ be sound approximations with $L^{\mdp{A}}_1,  L^{\mdp{A}}_2, U^{\mdp{A}}_1,  U^{\mdp{A}}_2 $ are all singleton sets with $1$, and $L^{\mdp{B}}_1 \defeq (0.001, 0.99),\, U^{\mdp{B}}_1 \defeq (0.009, 0.99)$. 
    Then, $\gapinf{L^{\mdp{A}}}{U^{\mdp{A}}} = 0$, $\gapinf{L^{\mdp{B}}}{U^{\mdp{B}}} = 0.008$. 
In the composition $\cmdp{L^\mdp{A}}\seqcomp\cmdp{L^\mdp{B}}$, we obtain a Pareto point on $\frac{0.001}{1-0.99} = 0.1$, while for $\cmdp{U^\mdp{A}}\seqcomp\cmdp{U^\mdp{B}}$, we obtain $\frac{0.009}{1-0.99} = 0.9$. Then, $\gapinf{L^{\mdp{A}\seqcomp \mdp{B}}}{U^{\mdp{A}\seqcomp \mdp{B}}} = 0.8$.
\end{example}
The example demonstrates that with highly-likely loops, errors measured in the infinity norm may be amplified. This motivates looking at different distance measures, e.g., based on ratios~\cite{DBLP:conf/fossacs/Chatterjee12}. These bounds may be tight for various shortcut MDPs, but require additional assumptions, such as inducing the same graph in the shortcut MDPs. We briefly discuss these bounds \iffull in~\cref{app:rem:bounds}\else in~\cite[Appendix D]{WVHRJ2024accepted}\fi.

Finally, the error may also disappear when composing oMDPs. For the motivating example in \cref{fig:motivatingexample}, a large error in room $B$ is irrelevant if the best scheduler never visits this room. We provide a concrete example:
\begin{example}
\label{ex:err_bounds_disappear}
Consider oMDPs $\mdp{A}$, $\mdp{B}$ in~\cref{fig:bidirectional_seqcomp_goodcase}. 
A lower bound $L$ for $\mdp{A}$ is a Pareto curve that only contains $(1,0)$, 
i.e., the point induced by taking action $a$. The error for this lower bound is $1$, as we may reach exit $\exrarg{2}$ with probability $1$. However, the error for $\cmdp{\fachievable{L}} \seqcomp \mdp{B}$ is $0$ as we already recover an optimal scheduler. 
\end{example}

\section{Implementation and Experiments}
\label{sec:imp_exp}

\subsubsection{Implementation}
We have implemented a prototypical \texttt{C++} extension of \textsc{\textsc{Storm}}~\cite{HenselJKQV22} that takes a string diagram in JSON-format as input. The syntax allows to name terms for simple reuse and oMDPs are defined as PRISM models with a list of entrances and exits defined via expressions over the variables. The tool supports caching of Pareto curves and shortcut MDPs, which is hugely beneficial if the same string diagram occurs in multiple contexts. We provide dedicated support for some syntactic sugar, most notably $n$-ary operations and the \emph{trace} operator~\cite{WatanabeEAH23}. We have implemented two approaches. 
The \emph{monolithic} (\texttt{Mon}) approach takes a string diagram $\sd{D}$ and inductively constructs the monolithic MDP $\semantics{\sd{D}}$. The \emph{recursive Pareto} computation (\texttt{rPareto}) follows the explanation in \cref{sec:naive_comp_algo}.


\subsubsection{Setup}
We run the algorithms using a time out of 15 minutes and a memory limit of 16GB. All experiments run on a single core of an Intel i9-10980XE processor.

\paragraph{Benchmarks}
Our benchmarks exhibit fundamental topological structures such as chains, which seems common structures for discretized (grid) worlds, and protocols that work in rounds or keep track of the number of rounds won/lost.  
Specifically, we create 8  benchmarks families with 50 different instances. We use three simple types of string diagrams: a two-dimensional grid of rooms with bidirectional doors (\texttt{BiGrid}), a grid of rooms that can only be passed in one direction (\texttt{UniGrid}), and a big chain (\texttt{Chain}). Each string diagram is initialized at the leaves with 6 to 16 different simple open MDPs of similar shape that occur multiple times. The shapes are a small room \texttt{RmS}, a big room \texttt{RmB}, and a selection of  biased dice \texttt{Dice}. Details are given \iffull in~\cref{app:allresults}\else in~\cite[Appendix A]{WVHRJ2024accepted}\fi. 

\newcolumntype{L}[1]{>{\raggedright\let\newline\\\arraybackslash\hspace{0pt}}m{#1}}
\newcolumntype{C}[1]{>{\centering\let\newline\\\arraybackslash\hspace{0pt}}m{#1}}
\newcolumntype{R}[1]{>{\raggedleft\let\newline\\\arraybackslash\hspace{0pt}}m{#1}}

\paragraph{Baselines.}
The only comparable compositional algorithm in the literature executes scheduler enumeration~\cite{WatanabeEAH23}. We approximate this using $\texttt{Prec} = \texttt{rPareto}(\eta = 0)$, i.e., with precise Pareto curve computations. While pure scheduler enumeration produces less overhead, it requires analyzing all schedulers, whereas  \texttt{Prec} only computes the Pareto-optimal DM schedulers. Scheduler enumeration over more than $10^{12}$ schedulers is completely infeasible. All benchmarks in our benchmark set have over $10^{32}$ schedulers. The monolithic algorithm is not optimised but uses mature data structures for sparse model construction, in particular for building the oMDPs. All algorithms use standard settings for MDP solving in \textsc{\textsc{Storm}}, in particular OVI with precision $10^{-4}$ and double arithmetic.

\newcommand{\showpgfmark}[1]{\tikz[baseline=-0.4ex]\node[mark size=0.7ex]{\pgfuseplotmark{#1}};}
\begin{figure}[t]
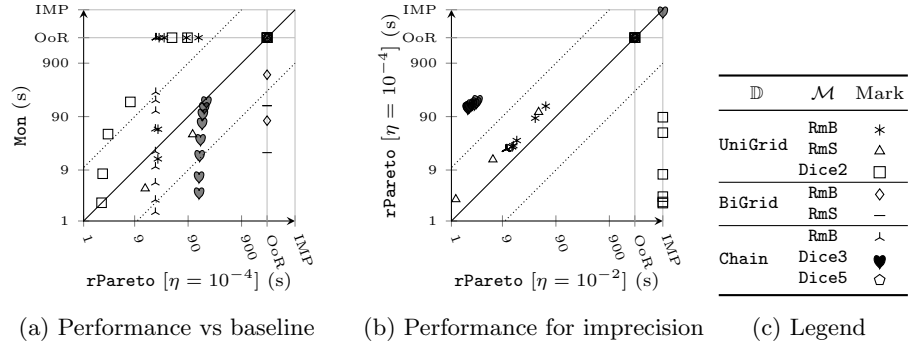

    \begin{subfigure}{.4\textwidth}
        \centering
        \scatterplot{plots/results_plot}{Total time Pareto4}{\texttt{rPareto} [$\eta = 10^{-4}$] (s)}{Total time Monolithic}{\texttt{Mon} (s)}
        \caption{Performance vs baseline}
        \label{fig:benchmark_scatter1}
    \end{subfigure}%
    \begin{subfigure}{.4\textwidth}
        \centering
        \scatterplot{plots/results_plot}{Total time Pareto2}{\texttt{rPareto} [$\eta = 10^{-2}$] (s)}{Total time Pareto4}{\texttt{rPareto} [$\eta = 10^{-4}$] (s)}
        \caption{Performance for imprecision}
        \label{fig:benchmark_scatter2}
    \end{subfigure}%
    \begin{subfigure}{.2\textwidth}%
{\scriptsize%
        \begin{NiceTabular}{ c c c }%
        \toprule
              $\sd{D}$ & $\mdp{M}$ & Mark \\\midrule
             \Block[L]{3-1}{\texttt{UniGrid}}  & \texttt{RmB}   & \showpgfmark{asterisk} \\
                                               & \texttt{RmS} & \showpgfmark{triangle} \\
                                               & \texttt{Dice2} & \showpgfmark{square} \\\hline
             \Block[L]{2-1}{\texttt{BiGrid}}   & \texttt{RmB}   & \showpgfmark{diamond} \\
                                               & \texttt{RmS} & \showpgfmark{-} \\\hline
             \Block[L]{3-1}{\texttt{Chain}}    & \texttt{RmB}   & \showpgfmark{Mercedes star} \\
                                               & \texttt{Dice3} & \showpgfmark{heart} \\
                                               & \texttt{Dice5} & \showpgfmark{pentagon} \\
             \bottomrule
        \end{NiceTabular}
}
    \caption{Legend}
    \label{tab:benchmark_series}
    \end{subfigure}
    \caption{Performance (time in s, OoR=time-out or memory-out, IMP=imprecise)}
\end{figure}

\subsubsection{Results}
In~\cref{fig:benchmark_scatter1} (log-log scale) we compare the run time of 
\texttt{rPareto} and \texttt{Mon}. Every point matches a benchmark, the point $(x,y)$ indicates that \texttt{Mon} required $x$ seconds, while \texttt{rPareto} took $y$ seconds. A point above the diagonal means that \texttt{rPareto} was faster. A point above the dotted diagonal means that \texttt{rPareto} was 10x faster. A point is on the IMP line, if the error is bigger then $0.8 \cdot 10^{-3}$, i.e., significantly above $10^{-4}$. \cref{fig:benchmark_scatter2} similarly compares \texttt{rPareto}($\eta = 10^{-4}$) vs \texttt{rPareto}($\eta = 10^{-2}$) to visualise the performance benefit when reducing $\eta$. \cref{tab:concrete_benchmarks} gives details: The columns give the string diagram and class of open MDPs, the number of monolithic states $|S_{\semantics{\sd{D}}}|$, the unique $\#\mdp{A}$ oMDPs in $\sd{D}$, and their state count $|S_\mdp{A}|$. We then consider the \texttt{Mon} and \texttt{rPareto} with $\eta \in \{0, 10^{-2}, 10^{-4}\}$ respectively: We list the total time $t$, the model building time $t_m$, the error $E$, and the total number of vertices $p$ in all sound approximations $(L,U)$ combined. For \texttt{Mon}, the error is guaranteed to be below $10^{-4}$. A complete table is \iffull in~\cref{app:allresults}\else in~\cite[Appendix A]{WVHRJ2024accepted}\fi.

\newcommand{\OOM}{MO}
\newcommand{\TO}{TO}
\begin{table}[t]
    \caption{Benchmark details (time in s, MO=memory out, TO=time out)}
\begin{adjustbox}{width=\textwidth}
    \centering
    \begin{NiceTabular}{l l r r r || r r | r | r r r r | r r r r }
        \toprule
             \Block[C]{2-1}{$\sd{D}$} & \Block[C]{2-1}{$\mdp{M}$} & \Block[C]{2-1}{$|S_{\semantics{\sd{D}}}|$} & \Block[C]{2-1}{$\#\mdp{A}$} & \Block[C]{2-1}{$|S_\mdp{A}|$} & \Block{1-2}{\texttt{Mon}} & & \texttt{Prec} & \Block{1-4}{\texttt{rPareto} [$\eta = 10^{-2}$]} & & & & \Block{1-4}{\texttt{rPareto} [$\eta = 10^{-4}$]} \\
             &  &  &  &  & $t$ & $t_m$ & $t$ & $t$ & $t_m$ & $E$ & \hspace{5mm}$p$ & $t$ & $t_m$ & $E$ & \hspace{5mm}$p$ \\\midrule
\texttt{UniGrid100} & \texttt{RmB} &1.1e+08&14&148582&\OOM&\OOM&\OOM&17&2&7e-11&58&32&2&4e-12&155\\
\texttt{UniGrid200} & \texttt{RmB} &4.2e+08&14&148582&\OOM&\OOM&\OOM&37&2&2e-19&58&84&2&1e-20&155\\
\texttt{UniGrid250} & \texttt{RmB} &6.6e+08&14&148582&\OOM&\OOM&\OOM&58&2&5e-23&58&141&2&3e-24&155\\
\texttt{UniGrid100} & \texttt{Dice}&5.4e+07&14&74424&170&41&\OOM&5&$<$1&3e-01&56&7&$<$1&2e-04&94\\
\texttt{UniGrid200} & \texttt{Dice}&2.1e+08&14&65424&\OOM&\OOM&\OOM&25&$<$1&8e-01&56&45&$<$1&5e-04&94\\
\texttt{UniGrid250} & \texttt{Dice}&3.3e+08&14&74424&\OOM&\OOM&\OOM&46&$<$1&1e+00&56&88&$<$1&7e-04&94\\
\texttt{Chain1000} &\texttt{RmB}&1.6e+07&6&42463&181&16&\OOM&11&$<$1&2e-50&44&22&$<$1&6e-52&82\\
\texttt{Chain2000} & \texttt{RmB}&2.1e+07&6&42463&258&24&\OOM&11&$<$1&3e-68&44&22&$<$1&7e-70&82\\
\texttt{BiGrid100} & \texttt{RmS}&8.5e+05&16&1353&19&1&\TO&\TO&\TO&\TO&\TO&\TO&\TO&\TO&\TO\\
\texttt{BiGrid200} & \texttt{RmS}&3.4e+06&16&1352&145&5&\TO&\TO&\TO&\TO&\TO&\TO&\TO&\TO&\TO\\
         \bottomrule
    \end{NiceTabular}
\end{adjustbox}
   \vspace{0.5cm} 
    \label{tab:concrete_benchmarks}
\end{table}

\subsubsection{Discussion} We discuss performance and error bounds. 

\paragraph{Performance.}
Our \texttt{Monolithic} baseline is reasonably fast, can construct a monolithic MDP with millions of states in a matter of seconds and analyze it in minutes. However, the approach reaches memory limitations when handling $>10^8$ states. 
In contrast, \texttt{rPareto} reduces the amount of time for model building, may speed up model checking by orders of magnitude, and is applicable even when the (monolithic) MDP has $>10^8$ states.  The `vertical lines' in \cref{fig:benchmark_scatter1} are due to the effect of caching intermediate results. However, \texttt{rPareto} is not a silver bullet: In particular, handling open MDPs with more than 3 exits, in particular as in \texttt{BiGrid}, is challenging to the multi-objective engine. This effect is amplified when using open MDPs that are challenging for value iteration, due to the iterative nature of \cref{alg:approxMultiObjMDP}. Precisely computing Pareto curves is not competitive at all for any of these benchmarks due to the large number of schedulers.

\paragraph{Error bounds.}
In all experiments, when $\eta=10^{-4}$, the approximations are sufficiently tight such that $E < 0.8 \cdot 10^{-3}$. The error actually mostly collapses: The lower bound typically includes a scheduler `close' to the optimal scheduler, similar to \cref{ex:err_bounds_disappear}. Consider~\cref{fig:benchmark_scatter2}: Increasing $\eta$ from $10^{-4}$ to $10^{-2}$ decreases the number of Pareto points roughly by a factor 2-3 and similarly the model checking time: However, the error sometimes explodes and for \texttt{UniGrid250}/\texttt{Dice}, the error is $1$.  



\section{Related Work and Conclusion}

\subsubsection{Related Work} Compositional verification methods for sequential MDPs~\cite{DBLP:conf/ijcai/BarryKL11,JungesS22,DBLP:conf/aips/NearyVCT22,WatanabeEAH23} have been discussed in~\cref{sec:intro}. For hierarchical MCs~\cite{DBLP:conf/qest/AbrahamJWKB10}, there are no schedulers. Hierarchical methods for reinforcement learning (RL) have been surveyed in \cite{DBLP:journals/csur/PateriaSTQ21}. Particularly interesting is~\cite{DBLP:conf/nips/JothimuruganBBA21}, which applies compositional RL w.r.t.\ to some sequentially composed MDP, derived from the task specification. 

Compositional probabilistic verification w.r.t.\ the parallel composition is investigated in~\cite{DBLP:conf/atva/XuGG10,DBLP:conf/atva/FengHKP11,DBLP:conf/tacas/KwiatkowskaNPQ10}. Most relevant is an multi-objective optimization framework~\cite{KwiatkowskaNPQ13} that reasons compositionally in an assume-guarantee fashion. 

The computation of multi-objective reward-bounded properties~\cite{DBLP:journals/jar/HartmannsJKQ20} generalizes topological value iteration. Their notion of episodes resembles rightward sequential composition. We support bidirectional composition where topological value iteration may not apply. For rightward MDPs, our approach caches the Pareto curves and thus supports exponentially many episodes in linear time.

\subsubsection{Conclusion}
In this work, we employ multi-objective model checking of monolithic MDPs to obtain a novel compositional algorithm for MDPs compositionally defined by string diagrams. Future work includes support for reward properties and a complete temporal logic, in particular also including support for finite horizon properties.  Furthermore, it is interesting to develop property-driven algorithms that compute Pareto-curves up to a context-dependent precision and to improve scalability for oMDPs with many exits. Finally, we think it is important to extract the compositional structure in form of string diagrams from popular formalisms like Prism programs. 

\clearpage
%
%
%
\bibliographystyle{splncs04}
\bibliography{TACAS2024}

\iffull 
\clearpage
\appendix

\section{Benchmark Details}\label{app:benchmarks}
All benchmark cases presented in \cref{sec:imp_exp} are formed by some combination of a string diagram $\sd{D}$ and leaf models $\mdp{M}$. 

\subsection{String Diagrams}
\paragraph{Unidirectional Grid (\texttt{UniGrid})}
The unidirectional grid string diagram contains a $N \times N$ grid of connected rooms. In the initial room at $(1,1)$, there is an exit to the north and the east, leading to $(2,1)$ and $(1,2)$ respectively. Once an edge of the grid has been reached, i.e., $(x, N)$ or $(N, y)$, only the east and north exit are available, respectively. The goal of the string diagram is to reach the unique exit located in $(N,N)$.

\paragraph{Bidirectional Grid (\texttt{BiGrid})}
The bidirectional grid is defined in the same way as the unidirectional grid, except that it also possible to traverse between the rooms in the opposite directions, south and west.


\paragraph{Unidirectional Chain with a loop (\texttt{Chain})}
The unidirectional chain string diagram can be seen as a 1D  version of the unidirectional grid. However, instead of having a unique exit at the end of the chain, there is another exit that brings us back to the start of the chain, which makes the chain not rightward. 
To define this string diagram, we use syntactic sugar in form of the \emph{trace} operator.


\subsection{Leaf Models}
\paragraph{Small Room (\texttt{RmS})}
There are four different variants of the small room model, indicated by a pair out of the set $\{\mathit{Safe}, \mathit{Unsafe}\} \times \{ \mathit{Windy}, \mathit{Calm} \}$. This pair determines the dynamics of the room.

The small room model consists of a $7 \times 7$ grid world with imprecise movement. After each movement action (north, east, south or west), there is some probability that the agent does not end up where it intended to move. This behaviour is more likely if the room is windy instead of calm. Furthermore, there are some holes in the grid, which cannot be exited once entered. The rooms that are unsafe contain more holes. The exits of the room are at the four center positions of each edge of the grid.

\paragraph{Big Room (\texttt{RmB})}
The big room model is defined in the same way as the small room, except that the dimensions of the grid are 101 instead of 7.

\paragraph{Dice Game (\texttt{Dice})}
The dice models contain a small dice game. The game is played in rounds, and the goal of the game is to score as many points as possible. In each round, there is a choice of three biased dice. After each round, the controller picks a die and throws it, and adds the (potentially negative) number on the die to their score. There are 100 rounds, and the score is clamped between 0 and 100. For the four exit variant of the dice game, obtaining a score between 0 and 24 means that the first exit will be taken. Similarly, a score between 25 and 49 means that the second exit will be taken, and so forth.



\subsection{Results and Models Table}
\label{app:allresults}


\begin{table}[H]
    \centering
    \begin{adjustbox}{scale=0.7}
    \begin{NiceTabular}{l l r r r || r r | r | r r r r | r r r r }
    \toprule
         \Block[C]{2-1}{$\sd{D}$} & \Block[C]{2-1}{$\mdp{M}$} & \Block[C]{2-1}{$|S_{\semantics{\sd{D}}}|$} & \Block[C]{2-1}{$\#\mdp{A}$} & \Block[C]{2-1}{$|S_\mdp{A}|$} & \Block{1-2}{\texttt{Mono}} & & \texttt{Prec} & \Block{1-4}{\texttt{rPareto} [$\eta = 10^{-2}$]} & & & & \Block{1-4}{\texttt{rPareto} [$\eta = 10^{-4}$]} \\
         &  &  &  &  & $t$ & $t_m$ & $t$ & $t$ & $t_m$ & $E$ & \hspace{5mm}$p$ & $t$ & $t_m$ & $E$ & \hspace{5mm}$p$ \\\midrule
\texttt{BiGrid10}&\texttt{RmB}&1.1e+06&16&169808&76&4&TO&TO&TO&TO&TO&TO&TO&TO&TO\\
\texttt{BiGrid50}&\texttt{RmB}&TO&TO&TO&TO&TO&TO&TO&TO&TO&TO&TO&TO&TO&TO\\
\texttt{BiGrid100}&\texttt{RmS}&8.5e+05&16&1353&19&1&TO&TO&TO&TO&TO&TO&TO&TO&TO\\
\texttt{BiGrid200}&\texttt{RmS}&3.4e+06&16&1352&145&5&TO&TO&TO&TO&TO&TO&TO&TO&TO\\
\texttt{BiGrid500}&\texttt{RmS}&3.4e+06&16&1352&144&5&TO&TO&TO&TO&TO&TO&TO&TO&TO\\
\texttt{Chain10}&\texttt{RmB}&1.1e+05&6&42463&1&$<$1&OOM&10&$<$1&2e-03&44&22&$<$1&8e-05&82\\
\texttt{Chain20}&\texttt{RmB}&2.1e+05&6&42463&2&$<$1&OOM&11&$<$1&3e-03&44&22&$<$1&1e-04&82\\
\texttt{Chain50}&\texttt{RmB}&5.3e+05&6&42463&5&1&OOM&10&$<$1&9e-04&44&22&$<$1&4e-05&82\\
\texttt{Chain100}&\texttt{RmB}&1.1e+06&6&42463&10&2&OOM&11&$<$1&3e-05&44&22&$<$1&2e-06&82\\
\texttt{Chain200}&\texttt{RmB}&2.1e+06&6&42463&20&3&OOM&11&$<$1&4e-08&44&22&$<$1&2e-09&82\\
\texttt{Chain500}&\texttt{RmB}&5.3e+06&6&42463&53&7&OOM&11&$<$1&5e-17&44&22&$<$1&2e-18&82\\
\texttt{Chain1000}&\texttt{RmB}&1.1e+07&6&42463&114&12&OOM&11&$<$1&1e-33&44&22&$<$1&5e-35&82\\
\texttt{Chain1500}&\texttt{RmB}&1.6e+07&6&42463&181&16&OOM&11&$<$1&2e-50&44&22&$<$1&6e-52&82\\
\texttt{Chain2000}&\texttt{RmB}&2.1e+07&6&42463&258&24&OOM&11&$<$1&3e-68&44&22&$<$1&7e-70&82\\
\texttt{Chain2500}&\texttt{RmB}&2.7e+07&6&42463&OOM&OOM&OOM&11&$<$1&2e-86&44&23&$<$1&3e-88&82\\
\texttt{Chain3000}&\texttt{RmB}&3.2e+07&6&42463&OOM&OOM&OOM&11&$<$1&1e-102&44&23&$<$1&2e-104&82\\
\texttt{Chain3500}&\texttt{RmB}&3.7e+07&6&42463&OOM&OOM&OOM&11&$<$1&4e-120&44&23&$<$1&5e-122&82\\
\texttt{Chain4000}&\texttt{RmB}&4.2e+07&6&42463&OOM&OOM&OOM&11&$<$1&7e-136&44&22&$<$1&8e-138&82\\
\texttt{Chain10}&\texttt{Dice3}&9.9e+04&3&9883&4&$<$1&OOM&2&$<$1&2e-16&211&145&$<$1&4e-16&1779\\
\texttt{Chain20}&\texttt{Dice3}&2.0e+05&3&9883&7&$<$1&OOM&2&$<$1&4e-16&211&144&$<$1&4e-16&1779\\
\texttt{Chain50}&\texttt{Dice3}&4.9e+05&3&9883&19&$<$1&OOM&2&$<$1&2e-15&211&148&$<$1&4e-16&1779\\
\texttt{Chain100}&\texttt{Dice3}&9.9e+05&3&9883&37&$<$1&OOM&2&$<$1&4e-15&209&154&$<$1&4e-16&1779\\
\texttt{Chain200}&\texttt{Dice3}&2.0e+06&3&9883&74&$<$1&OOM&2&$<$1&9e-15&211&165&$<$1&1e-15&1771\\
\texttt{Chain300}&\texttt{Dice3}&3.0e+06&3&9883&110&1&OOM&3&$<$1&1e-14&209&170&$<$1&1e-15&1777\\
\texttt{Chain400}&\texttt{Dice3}&3.9e+06&3&9883&146&2&OOM&3&$<$1&2e-14&213&181&$<$1&2e-15&1777\\
\texttt{Chain500}&\texttt{Dice3}&4.9e+06&3&9883&188&2&OOM&3&$<$1&2e-14&213&198&$<$1&2e-15&1767\\
\texttt{Chain100}&\texttt{Dice5}&OOM&OOM&OOM&OOM&OOM&OOM&OOM&OOM&OOM&OOM&OOM&OOM&OOM&OOM\\
\texttt{Chain200}&\texttt{Dice5}&OOM&OOM&OOM&OOM&OOM&OOM&OOM&OOM&OOM&OOM&OOM&OOM&OOM&OOM\\
\texttt{Chain500}&\texttt{Dice5}&OOM&OOM&OOM&OOM&OOM&OOM&OOM&OOM&OOM&OOM&OOM&OOM&OOM&OOM\\
\texttt{Chain10}&\texttt{RmS}&8.6e+02&6&349&$<$1&$<$1&OOM&$<$1&$<$1&4e-03&55&$<$1&$<$1&3e-04&133\\
\texttt{Chain20}&\texttt{RmS}&1.7e+03&6&349&$<$1&$<$1&OOM&$<$1&$<$1&1e-02&55&$<$1&$<$1&2e-04&133\\
\texttt{Chain50}&\texttt{RmS}&4.2e+03&6&349&$<$1&$<$1&OOM&$<$1&$<$1&6e-03&55&$<$1&$<$1&8e-05&133\\
\texttt{Chain100}&\texttt{RmS}&8.5e+03&6&349&$<$1&$<$1&OOM&$<$1&$<$1&4e-04&55&$<$1&$<$1&5e-06&133\\
\texttt{Chain200}&\texttt{RmS}&1.7e+04&6&349&$<$1&$<$1&OOM&$<$1&$<$1&1e-06&55&$<$1&$<$1&1e-08&133\\
\texttt{Chain500}&\texttt{RmS}&4.2e+04&6&349&$<$1&$<$1&OOM&$<$1&$<$1&2e-14&55&$<$1&$<$1&2e-16&133\\
\texttt{UniGrid10}&\texttt{RmB}&1.1e+06&14&148582&15&4&OOM&14&2&8e-04&58&24&2&4e-05&155\\
\texttt{UniGrid20}&\texttt{RmB}&4.2e+06&14&148582&52&10&OOM&14&2&2e-04&58&25&2&1e-05&155\\
\texttt{UniGrid50}&\texttt{RmB}&2.7e+07&14&148582&OOM&OOM&OOM&14&2&2e-06&58&26&2&1e-07&155\\
\texttt{UniGrid100}&\texttt{RmB}&1.1e+08&14&148582&OOM&OOM&OOM&17&2&7e-11&58&32&2&4e-12&155\\
\texttt{UniGrid200}&\texttt{RmB}&4.2e+08&14&148582&OOM&OOM&OOM&37&2&2e-19&58&84&2&1e-20&155\\
\texttt{UniGrid250}&\texttt{RmB}&6.6e+08&14&148582&OOM&OOM&OOM&58&2&5e-23&58&141&2&3e-24&155\\
\texttt{UniGrid500}&\texttt{RmB}&OOM&OOM&OOM&OOM&OOM&OOM&OOM&OOM&OOM&OOM&OOM&OOM&OOM&OOM\\
\texttt{UniGrid10}&\texttt{Dice2}&5.3e+05&14&74424&2&$<$1&OOM&2&$<$1&2e-02&56&2&$<$1&3e-05&94\\
\texttt{UniGrid20}&\texttt{Dice2}&2.3e+06&14&83424&8&2&OOM&2&$<$1&4e-02&56&2&$<$1&4e-05&94\\
\texttt{UniGrid50}&\texttt{Dice2}&1.4e+07&14&74424&42&10&OOM&2&$<$1&1e-01&56&3&$<$1&1e-04&94\\
\texttt{UniGrid100}&\texttt{Dice2}&5.4e+07&14&74424&170&41&OOM&5&$<$1&3e-01&56&7&$<$1&2e-04&94\\
\texttt{UniGrid200}&\texttt{Dice2}&2.1e+08&14&65424&OOM&OOM&OOM&25&$<$1&8e-01&56&45&$<$1&5e-04&94\\
\texttt{UniGrid250}&\texttt{Dice2}&3.3e+08&14&74424&OOM&OOM&OOM&46&$<$1&1e+00&56&88&$<$1&7e-04&94\\
\texttt{UniGrid300}&\texttt{Dice2}&OOM&OOM&OOM&OOM&OOM&OOM&OOM&OOM&OOM&OOM&OOM&OOM&OOM&OOM\\
\texttt{UniGrid350}&\texttt{Dice2}&OOM&OOM&OOM&OOM&OOM&OOM&OOM&OOM&OOM&OOM&OOM&OOM&OOM&OOM\\
\texttt{UniGrid400}&\texttt{Dice2}&OOM&OOM&OOM&OOM&OOM&OOM&OOM&OOM&OOM&OOM&OOM&OOM&OOM&OOM\\
\texttt{UniGrid450}&\texttt{Dice2}&OOM&OOM&OOM&OOM&OOM&OOM&OOM&OOM&OOM&OOM&OOM&OOM&OOM&OOM\\
\texttt{UniGrid500}&\texttt{Dice2}&OOM&OOM&OOM&OOM&OOM&OOM&OOM&OOM&OOM&OOM&OOM&OOM&OOM&OOM\\
\texttt{UniGrid10}&\texttt{RmS}&8.4e+03&14&1183&$<$1&$<$1&OOM&$<$1&$<$1&3e-03&108&$<$1&$<$1&1e-05&284\\
\texttt{UniGrid20}&\texttt{RmS}&3.4e+04&14&1184&$<$1&$<$1&OOM&$<$1&$<$1&1e-03&108&$<$1&$<$1&1e-05&284\\
\texttt{UniGrid50}&\texttt{RmS}&2.1e+05&14&1183&$<$1&$<$1&OOM&1&$<$1&6e-06&108&3&$<$1&6e-08&284\\
\texttt{UniGrid100}&\texttt{RmS}&8.5e+05&14&1183&4&$<$1&OOM&6&$<$1&7e-10&108&14&$<$1&2e-11&284\\
\texttt{UniGrid200}&\texttt{RmS}&3.4e+06&14&1182&42&3&OOM&43&$<$1&2e-18&108&109&$<$1&9e-21&284\\
\texttt{UniGrid500}&\texttt{RmS}&OOM&OOM&OOM&TO&TO&OOM&OOM&OOM&OOM&OOM&OOM&OOM&OOM&OOM\\
     \end{NiceTabular}
\end{adjustbox}
\end{table}


\clearpage
\section{Details on shortcut MDPs~(Section~\ref{sec:paretotoshortcut})}

\subsection{Proof for \cref{lem:addingSuboptimalActions}}
\label{app:proof:addingSuboptimalActions}
This corresponds to adding an additional action $a_{*}$ enabled in a state $s$, see \cref{fig:constructsuboptimal}. If the action is locally dominated, in the sense that the probabilities to the successor states are dominated by a weighted combination of existing actions, then this does not change the weighted reachability (as an optimal scheduler will not take these actions) and thus the Pareto curve.
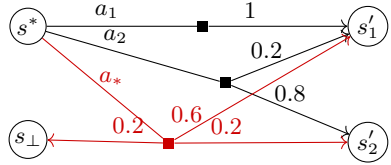
\begin{figure}[t]
\centering
\begin{tikzpicture}
    \node[circle,draw,inner sep=1pt] (s) {$s^{*}$};
    \node[circle,draw,inner sep=1pt, right=4cm of s] (sp1) {$s'_1$};
    \node[circle,draw,inner sep=1pt, below=of sp1] (sp2) {$s'_2$}; 
    \node[circle,draw,inner sep=1pt, below=of s] (sbot) {$s_\bot$};

    \node[fill,inner sep=2pt,right=2cm of s] (a1) {};
    \node[fill,inner sep=2pt,below=6mm of a1,xshift=3mm] (a2) {};
    \node[fill=red!80!black,inner sep=2pt,left=6mm of a2,yshift=-8mm] (a3) {};

    \draw[-] (s) -- node[above,pos=0.4] {$a_1$} (a1);
    \draw[-] (s) -- node[above,pos=0.4] {$a_2$} (a2);
    \draw[-,color=red!80!black] (s) -- node[right,pos=0.4] {$a_{*}$} (a3);
    \draw[->] (a1) -- node[above,pos=0.3] {$1$} (sp1);
    \draw[->] (a2) -- node[above,pos=0.3] {$0.2$} (sp1);
    \draw[->] (a2) -- node[above,pos=0.5] {$0.8$} (sp2);
    \draw[->,color=red!80!black] (a3) -- node[left,pos=0.25,xshift=-1mm] {$0.6$} (sp1);
    \draw[->,color=red!80!black] (a3) -- node[above,pos=0.3] {$0.2$} (sp2);
    \draw[->,color=red!80!black] (a3) -- node[above,pos=0.3] {$0.2$} (sbot);  
\end{tikzpicture}

\caption{Illustrative construction in \cref{lem:addingSuboptimalActions}. }
\label{fig:constructsuboptimal}
\end{figure}

The following lemma formulates this idea for any MDP and uses that if only probabilities to a sink are increased, the probability to reach any (other) state must be decreased. We provide a proof in \cref{app:proof:addingSuboptimalActions}.
\begin{lemma}
\label{lem:addingSuboptimalActions}
    Let $\mdp{A} = ((S, A, P), \interface)$ be an oMDP with $s_\bot$ a sink state and $s^{*} \in S$. 
    Let $\mdp{A}' = ((S, A \cup \{ a_{*} \}, P'), \interface)$ be an oMDP such that \[ P'(s,a) = \begin{cases} P(s,a) &  \text{ for all $s \in S, a \in A$,}\\ \delta & \text{if $s=s^*$, $a=a_{*}$,} \\ 
    \bot & \text{otherwise,} \end{cases} \] where $\delta \in \dist{S}$ such that for some weight vector  $w \in \weights{A}$ it holds that $\delta(s') \leq \sum w_a P(s_{*},a, s')$ for all $s' \in S \setminus \{ s_\bot \}$. Then: $ \fachievable{\mdp{A}} =\fachievable{\mdp{A}'}.$
\end{lemma}

We use the following auxiliary lemma.

\begin{lemma}
\label{lem:monotonicity_mc}
Let $M_1 \defeq (S, P_1),\, M_2 \defeq (S, P_2)$ be MCs, and $s\in S$ be a sink state for both $M_1$ and $M_2$. Suppose that $P_1(s_1, s_2) \leq P_2(s_1, s_2)$ for any $s_1, s_2\in S\backslash \{s\}$. Then, $\Reacha{M_1}{s_1}{s_2} \leq \Reacha{M_2}{s_1}{s_2}$ for any $s_1, s_2\in S\backslash \{s\}$. 
\end{lemma}

\begin{proof}[Proof of \cref{lem:addingSuboptimalActions}]
We prove $\achievable{\mdp{A}}{i} = \achievable{\mdp{A}'}{i}$.
Clearly, $\achievable{\mdp{A}}{i} \subseteq \achievable{\mdp{A}'}{i}$: The achievable points are the union over the achievable points of the schedulers. First, every scheduler  $\sched \in \Dmsched{\mdp{A}}$ is also a scheduler in $\Dmsched{\mdp{A}'}$. Second, the induced Markov chains under schedulers coincide $\indMC{\mdp{A}}{\sched} = \indMC{\mdp{A}'}{\sched}$, thus $\Reacha{\mdp{A},\sched}{i}{O} =\Reacha{\mdp{A'},\sched}{i}{O}$. 

It remains to show that $\achievable{\mdp{A}}{i} \supseteq \achievable{\mdp{A'}}{i}$.
Assume for contradiction that there exists a scheduler $\sched^{*} \in \Gmsched{\mdp{A'}}$ such that $\Reacha{\mdp{A'},\sched^{*}}{i}{O} \not\in \achievable{\mdp{A}}{i}$. 
\begin{itemize}
\item 
W.l.o.g., via \cref{lem:dmsuffices}, we can assume $\sched^{*}$ is DM.
\item 
W.l.o.g., we can assume $\sched^{*}$ is optimal for some weighted objective, that is, it optimizes some $\WReacha{\mdp{A}'}{\convw{w}}{i}$. We use $\sched^{r}$ to consider the optimal scheduler in $\mdp{A}$ for the same weight $\convw{w}$.
\item 
Furthermore, observe that $\sched^{*}(s_{*}) = a_{*}$. In particular, this follows from $\sched^{*} \not\in \Gmsched{\mdp{A}}$ as again, the induced Markov chains would then coincide.
\end{itemize}
Now we consider the Bellman characteristic~\cite[Thm 10.100]{Baier08} with respect to the reachability probability $\WReacha{\mdp{A}'}{\convw{w}}{i}$. Let $x_{s}$ denote the optimal (weighted) reachability with respect to weight $\convw{w}$ when starting in some state $s$; we have that 
\begin{align*}
x_{s^*} & =  \sum_{s'} P'(s^*,\sched^{*}(s^*))(s') \cdot x_{s'} = \sum_{s'} P'(s^*,a_{*})(s') \cdot x_{s'} \\ & =  \sum_{s'} \delta(s') \cdot x_{s'} \leq \sum_{s'} \sum_{a \in A}  w_a P'(s_{*},a)(s') x_{s'} \leq \max_{a \in A} \sum_{s'} P'(s_{*},a)(s') x_{s'}
\end{align*}
and 
\begin{align*}
x_{s^*}  & = \max_{a \in A \cup \{a_{*} \}} \sum_{s'} P'(s,a) \cdot x_{s'} > \sum_{s'} P'(s,\sched^r(s^{*})) \cdot x_{s'} = \max_{a \in A} \sum_{s'} P'(s_{*},a)(s') x_{s'}.
\end{align*}

\end{proof}

\subsection{Proof for \cref{prop:correct_abst_openMDPs}}
\label{app:proof:correct_abst_openMDP}
\begin{proof}[\cref{prop:correct_abst_openMDPs}]
We show (1) $\achievable{\cmdp{\fachievable{\mdp{A}}}}{i} =  \achievable{\mdp{A}}{i}$ and then (2) the inclusions.

\smallskip
\noindent\emph{1):}
$\achievable{\cmdp{\fachievable{\mdp{A}}}}{i} =  \achievable{\mdp{A}}{i}$:  because the set of achievable points is the down-ward convex closure of the set of Pareto points that are induced by DM schedulers. By~\cref{def:abstractMDP}, every DM scheduler in $\cmdp{\fachievable{\mdp{A}}}$ induces a Pareto point, which corresponds to a Pareto point in $\mdp{A}$ from $i$ that is induced by a DM scheduler. Conversely, every DM scheduler that induces a Pareto point from $i$ in $\mdp{A}$ corresponds to an action in $\cmdp{\fachievable{\mdp{A}}}$, which also corresponds to a DM scheduler in $\cmdp{\achievable{\mdp{A}}{i}}$. 

\smallskip\noindent\emph{2)}: We first show the following auxiliary lemma for $X, Y \subseteq \dist{O}^\en$ where $X(i), Y(i)$ are finitely generated downward-closed convex sets. 
\begin{align} X(i) \subseteq Y(i)\quad\text{ implies }\quad\achievable{\cmdp{X}}{i} \subseteq \achievable{\cmdp{Y}}{i}. \label{eqref:prooflemma} \end{align}
To see this holds, we start with $\cmdp{Y}$ and iteratively apply \cref{lem:addingSuboptimalActions} to add all actions from $\cmdp{X}$. We call this $\cmdp{Y}'$. The new actions are all suboptimal in the sense of \cref{lem:addingSuboptimalActions} and we obtain $\fachievable{\cmdp{Y}'} = \fachievable{\cmdp{Y}}$. Consequentially, $\achievable{\cmdp{Y}}{i} = \achievable{\cmdp{Y}'}{i}$. Now, notice that $\cmdp{X}$ is a subMDP of $\cmdp{Y}'$ in the sense that $\cmdp{X}$ contains a subset of the enabled actions of $\cmdp{Y}'$. By definition, $\achievable{\cmdp{X}}{i} \subseteq \achievable{\cmdp{Y}'}{i}$ and \eqref{eqref:prooflemma} follows. Now, we use \eqref{eqref:prooflemma} twice: We instantiate $X$ and $Y$ once with the lower bound and the Pareto curve, and once with the Pareto curve and the upper bound. From this, monotonicity follows directly.
\end{proof}

\subsection{Proof of \cref{lem:construct_deterministic_scheduler}}
\label{app:proof:construct_deterministic_scheduler}
\begin{proofs}
    We proof the first part, the second part follows directly from  \cref{def:abstractMDP}.
    We construct $\tau$ as follows: for any entrances $i\in \en^{\mdp{A}}$, there is a Pareto-optimal memoryless scheduler $\sigma_i$ such that $\Reacha{\mdp{A},\sigma_i}{i}{\ex^{\mdp{A}}} = \Reacha{\cmdp{\fachievable{\mdp{A}}},\sigma}{i}{\ex^{\mdp{A}}}$ by definition of $\cmdp{\fachievable{\mdp{A}}}$. For any finite path $\pi\defeq s_0\xrightarrow{a_0}\dots\xrightarrow{a_{n-1}}s_{n}  $ on $\mdp{A}$, we define $\tau(\pi) \defeq \sigma_i(s_n)$ if $s_0 = i \in \en^{\mdp{A}}$, and $\tau(\pi) = a$  otherwise, where $a$ is any fixed action. The scheduler $\tau$ indeed satisfies the condition because for any finite path that is from any entrance $i$, the scheduler $\tau$ behaves exactly same as $\sigma_i$. 
    \qed
\end{proofs}

\clearpage
\section{Details on Composing Shortcut MDPs (Section~\ref{sec:composingshortcuts})}
\label{app:explainmainprop}
In this Section, we continue our work to establish~\cref{prop:correct_seqcomp}.  We formalize our proof idea for  \[ \achievable{\cmdp{\fachievable{\mdp{A}}}\seqcomp\cmdp{\fachievable{\mdp{B}}}}{i}  =
    \achievable{\mdp{A}\seqcomp\mdp{B}}{i}.\] 
    
     Therefore,  we define the sequential composition $\sigma^{\mdp{A}}\seqcomp \sigma^{\mdp{B}}$ of memoryless schedulers $\sigma^{\mdp{A}}, \sigma^{\mdp{B}}$ on $\mdp{A}, \mdp{B}$ as a scheduler on $\mdp{A}\seqcomp\mdp{B}$ that behaves like $\sigma^{\mdp{A}}$ on states from $\mdp{A}$ and $\sigma^{\mdp{B}}$ on states from $\mdp{B}$.   
 
 For the inclusion from left to right, we construct a scheduler for $\mdp{A}\seqcomp\mdp{B}$ from some schedulers for $\cmdp{\fachievable{\mdp{A}}}$ and $\cmdp{\fachievable{\mdp{B}}}$ that preserves reachability problems:
\begin{lemma}
\label{lem:construct_deterministic_scheduler_seqcomp}
Let $\mdp{A}, \mdp{B}$ oMDPs such that $\mdp{A}\seqcomp\mdp{B}$ is well-defined. Let $\sigma^{\mdp{A}}$ and $\sigma^{\mdp{B}}$ be memoryless schedulers on $\cmdp{\fachievable{\mdp{A}}}$ and $ \cmdp{\fachievable{\mdp{B}}}$, respectively.
There is a scheduler $\tau$ on $\mdp{A}\seqcomp \mdp{B}$ such that $\Reacha{\mdp{A}\seqcomp \mdp{B},\tau}{i}{\ex^{\mdp{A}\seqcomp \mdp{B}}} = \Reacha{\cmdp{\fachievable{\mdp{A}}}\seqcomp \cmdp{\fachievable{\mdp{B}}},\sigma^{\mdp{A}}\seqcomp \sigma^{\mdp{B}}}{i}{\ex^{\mdp{A}\seqcomp \mdp{B}}}$ for any $i\in \en^{\mdp{A}\seqcomp \mdp{B}}$. 
\end{lemma}
\begin{proofs}
Every action in $\sched^\mdp{A}$ and $\sched^\mdp{B}$ correspond to a memoryless strategy for one of their entrances $i$. The scheduler $\tau$ stores the last entrance $i$ that it used to enter an MDP\footnote{It is good to realize that entrance states can also be reached from other states in the same oMDP.}. More precisely, $\tau$ tracks the first state reached after leaving an exit state or the initial state if no exit has been reached. The scheduler behaves then based on $\sigma^\mdp{A}$ in entrance $i$ or $\sigma^\mdp{B}$ in entrance $i$, respectively. The complete proof is below.
\end{proofs}

Continueing our discussion on the equality in \cref{prop:correct_seqcomp}: From right to left, we take one indirection more: We take a scheduler on $\mdp{A}\seqcomp\mdp{B}$, then obtain schedulers for $\cmdp{\fachievable{A}}$ and $\cmdp{\fachievable{B}}$ and show that by reconstructing a scheduler $\tau'$ on $\mdp{A}\seqcomp\mdp{B}$ with the same reachability probability, that this new scheduler dominates~$\tau$.

\subsection{Proof for \cref{lem:construct_deterministic_scheduler_seqcomp}}
\begin{proof}

By~\cref{lem:construct_deterministic_scheduler}, we can construct schedulers $\tau^{\mdp{A}}, \tau^{\mdp{B}}$ such that 
    \begin{align*}
        \Reacha{\mdp{A},\tau^{\mdp{A}}}{i^{\mdp{A}}}{\ex^{\mdp{A}}} &= \Reacha{\cmdp{\fachievable{\mdp{A}}},\sigma^{\mdp{A}}}{i^{\mdp{A}}}{\ex^{\mdp{A}}},\,\\
        \Reacha{\mdp{B},\tau^{\mdp{B}}}{i^{\mdp{B}}}{\ex^{\mdp{B}}} &= \Reacha{\cmdp{\fachievable{\mdp{B}}},\sigma^{\mdp{B}}}{i^{\mdp{B}}}{\ex^{\mdp{B}}},
    \end{align*}
    for any $i^{\mdp{A}}\in \en^{\mdp{A}},\, i^{\mdp{B}}\in \en^{\mdp{B}}$. 
    
    Then, we construct a scheduler $\tau$ on $\mdp{A}\seqcomp \mdp{B}$ below; the intuition is that $\tau$ mimicks either schedulers $\tau^{\mdp{A}}$ or $\tau^{\mdp{B}}$, and switches them when it reaches to states in $\exr^{\mdp{A}}\cup \exl^{\mdp{B}}$. 
    
    For any exit $s\in \exr^{\mdp{A}}\cup \exl^{\mdp{B}}$, 
    we write $\nxt(s)$ for the corresponding entrance 
    in $\enr^{\mdp{B}}\cup \enl^{\mdp{A}}$.
    Precisely, for any finite path $\pi\defeq s_0\xrightarrow{a_0}s_1\xrightarrow{a_{n-1}}\dots$, the scheduler $\tau(\pi)$ is given by 
    \begin{itemize}
      
        \item if $s_0 \in \en^{\mdp{A}}$ and there are no states $s$ in $\pi$ such that $s\in \exr^{\mdp{A}}\cup \exl^{\mdp{B}}$, then $\tau(\pi) \defeq \tau^{\mdp{A}}(\pi)$, 
        \item if $s_0 \in \en^{\mdp{B}}$ and there are no states $s$ in $\pi$ such that $s\in \exr^{\mdp{A}}\cup \exl^{\mdp{B}}$, then $\tau(\pi) \defeq \tau^{\mdp{B}}(\pi)$, 
        \item if $s_0\in \en^{\mdp{A}}\cup \en^{\mdp{B}}$ and there is a state $s_n$ such that $s_n$ is the last state that satsfies $s_n\in \exr^{\mdp{A}}\cup \exl^{\mdp{B}}$ in $\pi$, and $s_n$ is the last state in $\pi$, then $\tau(\pi) \defeq \nxt(s_n)$, 
        \item  if $s_0\in \en^{\mdp{A}}\cup \en^{\mdp{B}}$ and there is a state $s_n$ such that $s_n$ is the last state that satsfies $s_n\in \exr^{\mdp{A}}\cup \exl^{\mdp{B}}$ in $\pi$, and $s_n$ is not the last state in $\pi$, then $\tau(\pi) \defeq \tau^{\mdp{A}}(\pi')$ if $s_{n}\in \exl^{\mdp{B}}$, and $\tau(\pi) \defeq \tau^{\mdp{B}}(\pi')$ if $s_{n}\in \exr^{\mdp{A}}$, where $\pi'$ is the suffix sequence $s_{n+1}\xrightarrow{a_{n+1}}s_{n+2}\dots$ of $\pi$ (whose head $s_{n+1}$ is in $ \enr^{\mdp{B}}\cup \enl^{\mdp{A}}$). 
          \item if $s_0\not \in \en^{\mdp{A}}\cup \en^{\mdp{B}}$, then $\tau(\pi) \defeq a$, where $a$ is any fixed action,
    \end{itemize}

    Then, for any $i\in \en^{\mdp{A}}\cup  \en^{\mdp{B}}$ and $j\in \ex^{\mdp{A}\seqcomp \mdp{B}}$, we can simplify the equations from \cref{sec:seqcompcorrect} and substitute the reachability probability using $\tau$ by the reachability probability under a restriction of that scheduler. It is good to notice that the schedulers $\tau^\mdp{A}$ and $\tau^\mdp{B}$ are fixed and not dependent on $\tau$.
    \begin{align*}
        \Reacha{\mdp{A}\seqcomp \mdp{B}, \tau}{i}{j} = \begin{cases}\Reacha{\mdp{A},\tau^{\mdp{A}}}{i}{j} + \sum_{k\in \exr^{\mdp{A}}} \Reacha{\mdp{A},\tau^{\mdp{A}}}{i}{k}\cdot  \Reacha{\mdp{A}\seqcomp \mdp{B}, \tau}{\nxt(k)}{j} &\text{ if }i\in \en^{\mdp{A}},\\
       \Reacha{\mdp{B},\tau^{\mdp{B}}}{i}{j} + \sum_{k\in \exl^{\mdp{B}}} \Reacha{\mdp{B},\tau^{\mdp{B}}}{i}{k}\cdot  \Reacha{\mdp{A}\seqcomp \mdp{B}, \tau}{\nxt(k)}{j} &\text{ if } i\in \en^{\mdp{B}}.
        \end{cases}
    \end{align*}
    In fact, $\big(\Reacha{\mdp{A}\seqcomp \mdp{B}, \tau}{i}{j}\big)_{i\in \en^{\mdp{A}}\cup \en^{\mdp{B}}}$ is the unique solution of the following linear equation: 
    \begin{align*}
        x_i \defeq \begin{cases}
            0 &\text{ if }\Reacha{\mdp{A}\seqcomp \mdp{B}, \tau}{i}{j} = 0,\\
            \Reacha{\mdp{A},\tau^{\mdp{A}}}{i}{j} + \sum_{k\in \exr^{\mdp{A}}} \Reacha{\mdp{A},\tau^{\mdp{A}}}{i}{k}\cdot x_{\nxt(k)}&\text{ if } \Reacha{\mdp{A}\seqcomp \mdp{B}, \tau}{i}{j} > 0 \text{ and $i\in \en^{\mdp{A}}$,}\\
            \Reacha{\mdp{B},\tau^{\mdp{B}}}{i}{j} + \sum_{k\in \exl^{\mdp{B}}} \Reacha{\mdp{B},\tau^{\mdp{B}}}{i}{k}\cdot x_{\nxt(k)}&\text{ if } \Reacha{\mdp{A}\seqcomp \mdp{B}, \tau}{i}{j} > 0 \text{ and $i\in \en^{\mdp{B}}$.}
        \end{cases}
    \end{align*}

    On the other hand, $\big(\Reacha{\mdp{A}\seqcomp \mdp{B}, \tau}{i}{j}\big)_{i\in \en^{\mdp{A}}\cup \en^{\mdp{B}}}$ is also the unique solution of the following equivalent linear equation:

    \begin{align*}
        x_i \defeq \begin{cases}
            0 &\text{ if }\Reacha{\mdp{A}\seqcomp \mdp{B}, \tau}{i}{j} = 0,\\
\Reacha{\cmdp{\fachievable{\mdp{A}}},\sigma^{\mdp{A}}}{i}{j} + \sum_{k\in \exr^{\mdp{A}}} \Reacha{\cmdp{\fachievable{\mdp{A}}},\sigma^{\mdp{A}}}{i}{k}\cdot x_{\nxt(k)}&\text{ if } \Reacha{\mdp{A}\seqcomp \mdp{B}, \tau}{i}{j} > 0 \text{ and $i\in \en^{\mdp{A}}$,}\\
            \Reacha{\cmdp{\fachievable{\mdp{B}}},\sigma^{\mdp{B}}}{i}{j} + \sum_{k\in \exl^{\mdp{B}}} \Reacha{\cmdp{\fachievable{\mdp{B}}},\sigma^{\mdp{B}}}{i}{k}\cdot x_{\nxt(k)}&\text{ if } \Reacha{\mdp{A}\seqcomp \mdp{B}, \tau}{i}{j} > 0 \text{ and $i\in \en^{\mdp{B}}$.}
        \end{cases}
    \end{align*}

    Since $(\Reacha{\cmdp{\fachievable{\mdp{A}}}\seqcomp \cmdp{\fachievable{\mdp{B}}}, \sigma}{i}{j})_{i\in \en^{\mdp{A}}\cup \en^{\mdp{B}}}$ is also the unique solution of the above linear equation, we can conclude that $\Reacha{\mdp{A}\seqcomp \mdp{B}, \tau}{i}{j} = \Reacha{\cmdp{\fachievable{\mdp{A}}}\seqcomp \cmdp{\fachievable{\mdp{B}}}, \sigma}{i}{j}$.
    \qed
\end{proof}

\subsection{Proof of \cref{prop:correct_seqcomp}}
Note that an explanation is given in \cref{sec:seqcompcorrect} and (in the introduction of) \cref{app:explainmainprop}.
\begin{proof}[\cref{prop:correct_seqcomp}]
    We prove the statement in three parts.
    
    \medskip\noindent 
    Firstly, we prove (1) \[ \achievable{\cmdp{\fachievable{\mdp{A}}}\seqcomp\cmdp{\fachievable{\mdp{B}}}}{i}  \subseteq 
    \achievable{\mdp{A}\seqcomp\mdp{B}}{i}.\] 
    
    From a DM scheduler $\sigma$ on $ \cmdp{\fachievable{\mdp{A}}}\seqcomp\cmdp{\fachievable{\mdp{B}}}$, we obtain  DM schedulers $\sigma^{\mdp{A}}, \sigma^{\mdp{B}}$ with $\sigma= \sigma^{\mdp{A}}\seqcomp \sigma^{\mdp{B}}$. 
    Using~\cref{lem:construct_deterministic_scheduler_seqcomp}, there exists a deterministic scheduler $\tau$ such that \[ \Reacha{\cmdp{\fachievable{\mdp{A}}}\seqcomp\cmdp{\fachievable{\mdp{B}}}, \sigma}{i}{\ex^{\mdp{A}\seqcomp \mdp{B}}} = \Reacha{\mdp{A}\seqcomp \mdp{B}, \tau}{i}{\ex^{\mdp{A}\seqcomp \mdp{B}}}. \] 

    \medskip
    \noindent Next, we prove (2) \[   
    \achievable{\mdp{A}\seqcomp\mdp{B}}{i}  \subseteq \achievable{\cmdp{\fachievable{\mdp{A}}}\seqcomp\cmdp{\fachievable{\mdp{B}}}}{i} .\]
   
    For each DM scheduler  $\tau$ on $ \mdp{A}\seqcomp \mdp{B}$, it induces DM schedulers $\tau^{\mdp{A}}, \tau^{\mdp{B}}$ on $\mdp{A}, \mdp{B}$ such that $\tau $ such that $\tau = \tau^{\mdp{A}}\seqcomp \tau^{\mdp{B}}$.  
    We consider memoryless schedulers $\sigma^{\mdp{A}}, \sigma^{\mdp{B}}$ on $\cmdp{\fachievable{\mdp{A}}}, \cmdp{\fachievable{\mdp{B}}}$ such that $\Reacha{\mdp{A},\tau^{\mdp{A}}}{i^{\mdp{A}}}{\ex^\mdp{A}} \preceq \Reacha{\cmdp{\fachievable{\mdp{A}}},\sigma^{\mdp{A}}}{i^{\mdp{A}}}{\ex^{\mdp{A}}}$ for any $i^{\mdp{A}}\in \en^{\mdp{A}}$, and same for $\tau^{\mdp{B}}, \sigma^{\mdp{B}}$. These exist due to \cref{lem:recoverstratfromshortcut}.

By~\cref{lem:construct_deterministic_scheduler_seqcomp}, there is a scheduler $\tau'$ on $\mdp{A}\seqcomp \mdp{B}$ such that \[ \Reacha{\mdp{A}\seqcomp \mdp{B},\tau'}{i}{\ex^{\mdp{A}\seqcomp \mdp{B}}} = \Reacha{\cmdp{\fachievable{\mdp{A}}}\seqcomp \cmdp{\fachievable{\mdp{B}}},\sigma^{\mdp{A}}\seqcomp \sigma^{\mdp{B}}}{i}{\ex^{\mdp{A}\seqcomp \mdp{B}}}. \]

Then, by considering linear equations similar to these shown in the proof of~\cref{lem:construct_deterministic_scheduler_seqcomp}, we can obtain the inequality $\Reacha{\mdp{A}\seqcomp \mdp{B}, \tau}{i}{\ex^{\mdp{A}\seqcomp \mdp{B}}} \preceq \Reacha{\mdp{A}\seqcomp \mdp{B}, \tau'}{i}{\ex^{\mdp{A}\seqcomp \mdp{B}}}$, which concludes $\Reacha{\mdp{A}\seqcomp \mdp{B},\tau}{i}{\ex^{\mdp{A}\seqcomp \mdp{B}}} \preceq \Reacha{\cmdp{\fachievable{\mdp{A}}}\seqcomp \cmdp{\fachievable{\mdp{B}}},\sigma^{\mdp{A}}\seqcomp\sigma^{\mdp{B}}}{i}{\ex^{\mdp{A}\seqcomp \mdp{B}}}$. 


Finally, we prove the (3) the inclusions w.r.t.\ lower and upper bounds.
These are exactly analogously to the proof of \cref{prop:correct_abst_openMDPs} by adding suboptimal actions as in \cref{lem:addingSuboptimalActions} to the upper bound until the lower bound is a strict subMDP of the newly created MDP.  \qed
\end{proof}

\clearpage
\section{Details on Error Bounds in~\cref{sec:comp_error_bouds}}

\subsection{Proof for~\cref{thm:error_uni_seq}}
\label{subsec:proof_error_uni_seq}

We use the following propositions for the proof: 
\begin{proposition}
\label{prop:error_uni_seq}
    Let $\mdp{A}, \mdp{B}$ be oMDPs with $\typemdp{\mdp{A}} = (m,  0) \rightarrow (l, 0)$ and $\typemdp{\mdp{B}} = (l, 0) \rightarrow (n, 0)$, $i\in \en^{\mdp{A}}$, and $(L^{\mdp{A}}, U^{\mdp{A}}), (L^{\mdp{B}}, U^{\mdp{B}})$ be sound approximations. Assume the following: 
    \begin{itemize}
        \item $\sigma^{\mdp{A}},\sigma^{\mdp{B}}$ be memoryless schedulers on $\cmdp{U^{\mdp{A}}},\cmdp{U^{\mdp{B}}}$, and $\tau^{\mdp{A}},\tau^{\mdp{B}}$ be memoryless schedulers on $\cmdp{L^{\mdp{A}}},\cmdp{L^{\mdp{B}}}$,
        \item  $\epsilon^{\mdp{A}}$ be a constant given by
        \begin{align*}
            \epsilon^{\mdp{A}} \defeq \sup_{k\in \ex^{\mdp{A}}}|\Reacha{\cmdp{U^{\mdp{A}}}, \sigma^{\mdp{A}}}{i}{k} -  \Reacha{\cmdp{L^{\mdp{A}}}, \tau^{\mdp{A}}}{i}{k}|,
        \end{align*}
        \item $\epsilon^{\mdp{B}}$ be a constant given by
        \begin{align*}
            \epsilon^{\mdp{B}} \defeq \sup_{k\in \en^{\mdp{B}}} \sup_{j\in \ex^{\mdp{B}}}|\Reacha{\cmdp{U^{\mdp{B}}}, \sigma^{\mdp{B}}}{k}{j} -  \Reacha{\cmdp{L^{\mdp{B}}}, \tau^{\mdp{B}}}{k}{j}|,
        \end{align*}
 
    \end{itemize}
    Then, the following inequality holds:
    \begin{align*}
        \sup_{j\in \ex^{\mdp{B}}}|\Reacha{\cmdp{U^{\mdp{A}}}\seqcomp \cmdp{U^{\mdp{B}}}, \sigma^{\mdp{A}}\seqcomp \sigma^{\mdp{B}}}{i}{j} -  \Reacha{\cmdp{L^{\mdp{A}}}\seqcomp \cmdp{L^{\mdp{B}}}, \tau^{\mdp{A}}\seqcomp \tau^{\mdp{B}}}{i}{j}| \leq |\ex^{\mdp{A}}|\cdot \epsilon^{\mdp{A}} + \epsilon^{\mdp{B}}. 
    \end{align*}
\end{proposition}
\begin{proof}
    By the following calculation: 
   \begin{align*}
       &\sup_{j\in \ex^{\mdp{B}}}\big| \Reacha{\cmdp{U^{\mdp{A}}}\seqcomp \cmdp{U^{\mdp{B}}}, \sigma^{\mdp{A}}\seqcomp \sigma^{\mdp{B}}}{i}{j} -  \Reacha{\cmdp{L^{\mdp{A}}}\seqcomp \cmdp{L^{\mdp{B}}}, \tau^{\mdp{A}}\seqcomp \tau^{\mdp{B}}}{i}{j} \big|\\
       =  &\sup_{j\in \ex^{\mdp{B}}}\big| \Reacha{\cmdp{U^{\mdp{A}}}\seqcomp \cmdp{U^{\mdp{B}}}, \sigma^{\mdp{A}}\seqcomp \sigma^{\mdp{B}}}{i}{j} - \Reacha{\cmdp{U^{\mdp{A}}}\seqcomp \cmdp{L^{\mdp{B}}}, \sigma^{\mdp{A}}\seqcomp \tau^{\mdp{B}}}{i}{j}\\
       +& \Reacha{\cmdp{U^{\mdp{A}}}\seqcomp \cmdp{L^{\mdp{B}}}, \sigma^{\mdp{A}}\seqcomp \tau^{\mdp{B}}}{i}{j}- \Reacha{\cmdp{L^{\mdp{A}}}\seqcomp \cmdp{L^{\mdp{B}}}, \tau^{\mdp{A}}\seqcomp \tau^{\mdp{B}}}{i}{j} \big|\\
       = & \sup_{j\in \ex^{\mdp{B}}}\big| \sum_{k\in \ex^{\mdp{A}}} \Reacha{\cmdp{U^{\mdp{A}}}, \sigma^{\mdp{A}}}{i}{k} \cdot \big(\Reacha{\cmdp{U^{\mdp{B}}},\sigma^{\mdp{B}}}{k}{j} - \Reacha{\cmdp{L^{\mdp{B}}},\tau^{\mdp{B}}}{k}{j}\big)  \\
       + & \sum_{k\in \ex^{\mdp{A}}} \big(\Reacha{\cmdp{U^{\mdp{A}}},\sigma^{\mdp{A}}}{i}{k} - \Reacha{\cmdp{L^{\mdp{A}}},\tau^{\mdp{A}}}{i}{k}\big) \cdot \Reacha{\cmdp{L^{\mdp{B}}}, \tau^{\mdp{B}}}{k}{j}  \big|\\
       \leq& \sum_{k\in \ex^{\mdp{A}}} \epsilon^{\mdp{B}}\cdot \Reacha{\cmdp{U^{\mdp{A}}}, \sigma^{\mdp{A}}}{i}{k} + \epsilon^{\mdp{A}}\cdot \sup_{j\in \ex^{\mdp{B}}}\big(\Reacha{\cmdp{L^{\mdp{B}}}, \tau^{\mdp{B}}}{k}{j}\big)\\
       \leq & |\ex^{\mdp{A}}| \cdot \epsilon^{\mdp{A}} + \epsilon^{\mdp{B}} .
   \end{align*}
    \qed
\end{proof}

\begin{proposition}
\label{prop:error_accurate_uni_seq}
Let $\mdp{A}, \mdp{B}$ be oMDPs with $\typemdp{\mdp{A}} = (m,  0) \rightarrow (l, 0)$ and $\typemdp{\mdp{B}} = (l, 0) \rightarrow (n, 0)$, and $(L^{\mdp{A}}, U^{\mdp{A}}), (L^{\mdp{B}}, U^{\mdp{B}})$ be sound approximations.
The following inequality holds: 
\begin{align*}
    \gapinf{L}{U} \leq |\ex^{\mdp{A}}|\cdot \gapinf{L^{\mdp{A}}}{U^{\mdp{A}}} + \gapinf{L^{\mdp{B}}}{U^{\mdp{B}}}, 
\end{align*}
where $L\defeq\big( \achievable{\cmdp{L^{\mdp{A}}}\seqcomp \cmdp{L^{\mdp{B}}}}{i}\big)_{i\in \en^{\mdp{A}\seqcomp \mdp{B}}}$, and $U \defeq \big( \achievable{\cmdp{U^{\mdp{A}}}\seqcomp \cmdp{U^{\mdp{B}}}}{i}\big)_{i\in \en^{\mdp{A}\seqcomp \mdp{B}}}$.
\end{proposition}
\begin{proof}
Let $i\in \en^{\mdp{A}}$, and $\sigma = \sigma^{\mdp{A}}\seqcomp \sigma^{\mdp{B}}$ be a Pareto-optimal memoryless scheduler on $\cmdp{U^{\mdp{A}}}\seqcomp \cmdp{U^{\mdp{B}}}$. By definition of $\cmdp{\_}$, $\sigma^{\mdp{A}}, \sigma^{\mdp{B}}$ are Pareto-optimal for each entrance on $ \cmdp{U^{\mdp{A}}}, \cmdp{U^{\mdp{B}}}$, respectively. We take Pareto-optimal memoryless schedulers $\tau^{\mdp{A}}, \tau^{\mdp{B}}$ on $\cmdp{L^{\mdp{A}}}, \cmdp{L^{\mdp{B}}}$ such that 
\begin{align*}
     \sup_{k\in \ex^{\mdp{A}}}|\Reacha{\cmdp{U^{\mdp{A}}}, \sigma^{\mdp{A}}}{i}{k} -  \Reacha{\cmdp{L^{\mdp{A}}}, \tau^{\mdp{A}}}{i}{k}| &\leq  \gapinf{L^{\mdp{A}}}{U^{\mdp{A}}},\\
     \sup_{k\in \en^{\mdp{B}}}\sup_{j\in \ex^{\mdp{B}}}|\Reacha{\cmdp{U^{\mdp{B}}}, \sigma^{\mdp{B}}}{k}{j} -  \Reacha{\cmdp{L^{\mdp{B}}}, \tau^{\mdp{B}}}{k}{j}| &\leq  \gapinf{L^{\mdp{B}}}{U^{\mdp{B}}}. 
\end{align*}
Then, we obtain the following inequality by~\cref{prop:error_accurate_uni_seq}:
\begin{align*}
    &\sup_{j\in \ex^{\mdp{B}}}\big| \Reacha{\cmdp{U^{\mdp{A}}}\seqcomp \cmdp{U^{\mdp{B}}}, \sigma^{\mdp{A}}\seqcomp \sigma^{\mdp{B}}}{i}{j} -  \Reacha{\cmdp{L^{\mdp{A}}}\seqcomp \cmdp{L^{\mdp{B}}}, \tau^{\mdp{A}}\seqcomp \tau^{\mdp{B}}}{i}{j} \big| \\
    &\leq |\ex^{\mdp{A}}|\cdot \gapinf{L^{\mdp{A}}}{U^{\mdp{A}}} + \gapinf{L^{\mdp{B}}}{U^{\mdp{B}}}.
\end{align*}
By taking the infimum and supremum, the following inequality holds: 
\begin{align*}
    \gapinf{L}{U} \leq |\ex^{\mdp{A}}|\cdot \gapinf{L^{\mdp{A}}}{U^{\mdp{A}}} + \gapinf{L^{\mdp{B}}}{U^{\mdp{B}}}.
\end{align*}
\qed
\end{proof}

\begin{proposition}
\label{prop:triangle_inequality}
    Let $(L, B),\, (B, U)$ be sound approximations whose types are all same. The following triangle inequality holds: 
    \begin{align*}
        \gapinf{L}{U} \leq \gapinf{L}{B} + \gapinf{B}{U}. 
    \end{align*}
\end{proposition}
\begin{proof}
For any $i\in \en$ and $\point^{U}\in U_i$, we take $p(\point^{U})\in B_i$ such that $\|\point^{U}-p(\point^{U})\|_{\infty} =  \inf_{\point^{B}\in B_i}\|\point^{U} - \point^B  \|_{\infty}$. Then, 
\begin{align*}
    \sup_{\point^U\in U_i}\inf_{\point^{L} \in L_i} \| \point^U -  \point^L \|_{\infty} &= \sup_{\point^U\in U_i}\inf_{\point^{L} \in L_i}\| \point^U - p(\point^U) + p(\point^U) -  \point^L   \|_{\infty} \\  
    &\leq \sup_{\point^U\in U_i}\| \point^U - p(\point^U) \|_{\infty} +  \sup_{\point^U\in U_i}\inf_{\point^{L} \in L_i}\|p(\point^U) -  \point^L   \|_{\infty}\\
    &\leq \sup_{\point^U\in U_i} \inf_{\point^B\in B_i}\| \point^U - \point^B \|_{\infty} + \sup_{\point^B\in B_i}\inf_{\point^{L} \in L_i}\| \point^B -  \point^L   \|_{\infty}
\end{align*}\qed
\end{proof}

\begin{proof}[\cref{thm:error_uni_seq}]
By~\cref{prop:error_accurate_uni_seq,prop:triangle_inequality}. \qed
\end{proof}

\subsection{Remarks on Bounds for Composition}
\label{app:rem:bounds}

The example demonstrates that in the presence of likely loops, reachability probabilities can be rescaled, which may amplify errors measured in the infinity norm. This motivates looking at different distance measures. Indeed, error bounds on reachability probabilities in~\cite{DBLP:conf/fossacs/Chatterjee12} may yield tighter bounds. In a nutshell, for \emph{structurally equivalent MCs} $\mdp{C}, \mdp{D}$, i.e., MCs that share the same underlying graph, they derive a bound $|\Reacha{\mdp{C}}{i}{j} - \Reacha{\mdp{D}}{i}{j}| \leq \epsilon^{2\cdot |S|} - 1$ with $\epsilon\defeq \max_{s, s'\in S}\max\big(\frac{P^{\mdp{C}}(s, s')}{P^{\mdp{D}}(s, s')},\, \frac{P^{\mdp{D}}(s, s')}{P^{\mdp{C}}(s, s')}\big)$. This bound may be applied to Pareto-optimal schedulers $\sigma$ on $\cmdp{L^{\mdp{A}}}\seqcomp \cmdp{L^{\mdp{B}}}$ and $\tau$ on $\cmdp{U^{\mdp{A}}}\seqcomp \cmdp{U^{\mdp{B}}}$, if their induced MCs are indeed structurally equivalent. In particular, for these shortcut MDPs, $|S|$ is often small. However, to derive a-priori error bounds needs careful algorithmic considerations to ensure structural equivalence throughout all approximations. 

\ifdraft

\else 
\fi
\else
\fi
\end{document}